\RequirePackage{ifpdf}
\ifpdf
\documentclass{IEEEtran}
\else
\documentclass[dvips]{IEEEtran}
\fi

\usepackage[british]{babel}
\usepackage{amsmath,amsfonts,enumerate}
\usepackage[scaled]{helvet} 
\usepackage{latexsym,color}
\usepackage{notation}
\usepackage{bm}
\usepackage{algorithm,algorithmic,subfig}
\usepackage[all]{xy}
\usepackage{hyperref} 
\usepackage{graphicx}
\UseCrayolaColors

\CompileMatrices 

\newcommand{\mw}{w}

\floatstyle{ruled}
\restylefloat{algorithm}
\restylefloat{table}
\restylefloat{figure}

\newtheorem{theorem}{Theorem}
\newtheorem{lemma}[theorem]{Lemma}
\newtheorem{corollary}[theorem]{Corollary}
\newtheorem{definition}[theorem]{Definition}
\newtheorem{example}{Example}

\def\ehmm#1[#2]{\textsc{\MakeLowercase{#1}}_{#2}}

\newcommand{\Q}{{\ehmm{Q}[]}}
\newcommand{\R}{{\ehmm{R}[]}}

\newcommand{\qedex}{\hfill \IEEEQEDopen}
\newcommand{\pit}{\tau}

\newcommand{\ints}{\mathbb{Z}}
\newcommand{\posints}{\mathbb{Z}_+}
\newcommand{\nats}{\mathbb{N}}
\newcommand{\reals}{\mathbb{R}}

\newcommand{\xspace}{{\ensuremath{\cal X}}} 
\let\rv\bm

\newcommand{\prd}{\textnormal{\tiny p}}
\newcommand{\sil}{\textnormal{\tiny s}}

\newcommand{\name}[2][+<0pt,+10pt>]{\save[]#1*\txt{\tiny$\tuple{\text{#2}}$}\restore}
\newcommand{\bname}[1]{\name[+<0pt,-10pt>]{#1}}

\newcommand{\corref}[1]{Corollary~\ref{#1}}

\newcommand{\wstates}{Q}
\newcommand{\wstate}{q}
\newcommand{\wstart}{\wstate_0} 

\newcommand{\wtf}{\Prob_{\kern-0.3em\raisebox{-0.1ex}{\tiny$\vec{~}$}}}
\DeclareMathOperator{\wnl}{\Lambda} 

\DeclareMathOperator{\Prob}{P}

\DeclareMathOperator{\Exp}{E} 
\DeclareMathOperator{\kld}{D} 
\DeclareMathOperator{\ent}{H} 

\newcommand{\indicator}[1]{\mathbb I_{#1}}

\newcommand{\argmax}{\mathop{\textnormal{argmax}}}

\newcommand{\nsw}{\textnormal{\textsf{n}}}
\newcommand{\swi}{\textnormal{\textsf{s}}}

\newcommand{\expname}[1]{\textsc{#1}}
\newcommand{\ex}[1]{*+[o][F]{\expname{\small#1}}}
\newcommand{\si}{*+<4pt>[o][F**:Black]{}} 
\newcommand{\bn}{*+<4pt>[o][F]{}} 

\begin{document}

\author{Wouter~M.~Koolen and~Steven~de~Rooij
\thanks{S. de Rooij is with the 
Informatics Institute, University of Amsterdam
Science Park 904, P.O. Box 94323
1090 GH Amsterdam, Netherlands
(steven.de.rooij@gmail.com)
and with the 
VU University Amsterdam, 
De Boelelaan 1081a,
1081HV Amsterdam,
Netherlands 
.}
\thanks{W. M. Koolen is with the
Centrum Wiskunde \& Informatica (CWI),
  P.O. Box 94079, NL-1090 GB Amsterdam, Netherlands
(wmkoolen@cwi.nl)}
\thanks{A pre-print of this paper appeared as \cite{us:ceae:corr}. An extended abstract \cite{us:ceae:colt} of this paper appeared at COLT 2008. The dissertation of the first author includes this paper as \cite[Chapter 3]{dissertation}.}
\thanks{%
Copyright \copyright{} 2012 IEEE. Personal use of this material is permitted.  However, permission to use this material for any other purposes must be obtained from the IEEE by sending a request to pubs-permissions@ieee.org.
}
}

\title{Universal Codes from Switching Strategies}

\maketitle

\begin{abstract}
  We discuss algorithms for combining sequential prediction
  strategies, a task which can be viewed as a natural generalisation
  of the concept of universal coding. We describe a graphical language
  based on Hidden Markov Models for defining prediction strategies,
  and we provide both existing and new models as examples. The 
  models include efficient, parameterless models for switching between
  the input strategies over time, including a model for the case where
  switches tend to occur in clusters, and finally a new model for the
  scenario where the prediction strategies have a known relationship,
  and where jumps are typically between strongly related ones. This
  last model is relevant for coding time series data where parameter
  drift is expected. As theoretical contributions we introduce an
  interpolation construction that is useful in the development and
  analysis of new algorithms, and we establish a new sophisticated lemma
  for analysing the individual sequence regret of parameterised
  models.
\end{abstract}


\section{Introduction}\label{sec:intro}
For the most delectable universal codes, fill a cooking pot with
water, add a bunch of experts, a.k.a. codes, and put it on a slow
fire, stirring constantly. The resulting mix is guaranteed to delight,
achieving a codelength close to that of the best among the ingredient
codes.

In this paper we investigate such cookery in detail, deviating from
the usual recipe in two ways. First, following
Shtarkov~\cite{shtarkov.universal} and Rissanen~\cite{Rissanen1996},
universality is expressed in terms of the individual sequence regret:
the difference between the length of the considered code and the
shortest codelength among any of the ingredient codes, for the data
that were actually observed. As such, there are no distributional
assumptions. Second, the setting is generalised somewhat: rather than
always comparing our performance to that of the best code, we will
also consider \emph{combinations} of the ingredient codes as baselines
for the regret.  Among other things, this allows us to compete with
the best possible way to split the data sequence into a small number
of blocks of consecutive outcomes, and encode each block with the best
original code for that block, a problem known as \emph{expert
  tracking} in online learning
\cite{HerbsterWarmuth1998,cesa-bianchi2006}, which is also a core
focus of this work.

We identify sequential coding with sequential prediction as
follows. In each round $t=1,2\ldots$, a sequential prediction strategy
issues a probability distribution $Q_t$ on the outcome space
$\xspace$, which for simplicity we assume to be
countable. Subsequently, a new outcome $x_t$ is observed, and the
prediction is evaluated using \emph{logarithmic loss} $-\log_2
Q_t(x_t)$. The Kraft inequality states that there exists a prefix code
such that the accumulated logarithmic loss of the prediction strategy,
rounded up to the nearest integer, corresponds exactly to the
codelength for the data; in fact there are many practical algorithms
for implementing such a code, such as arithmetic coding
\cite{Rissanen1976}. For this reason, we will use the words
``prediction strategy'' and ``code'' interchangeably; similarly we use
``codelength'' as a synonym for ``logarithmic loss'', and forget about
the rounding. Also, for convenience we will use natural logarithms
such that codelengths are expressed in
nats. See~\cite{rissanen1984,barron1998b,MerhavFeder1998} for more
information about the connection between sequential prediction and
data compression.

Given a set of codes, our aim is to build universal codes with
efficient implementations, and evaluate their performance. Borrowing
terminology from learning theory, we henceforth call the ingredient
codes ``experts'' to emphasize that they are black boxes that can be
interpreted as prediction strategies. The overall protocol is as
follows (see \figref{fig:pwea.protocl}). Let $\Xi$ be a set of
experts, that we fix throughout this paper. Each round $t$, each
expert $\xi\in\Xi$ issues a prediction $P_{\xi,t}$ of the next outcome
in the form of a probability distribution. Our universal prediction
strategy collects all these predictions and uses them to form a
prediction $Q_t$ of its own. We then wait until the new data item
$x_t$ is observed, and incur a logarithmic loss of (encode it using)
$-\ln Q_t(x_t)$ nats. By the end of the game we compare our
accumulated loss (total codelength) to that of the best among a set of
reference strategies.

\begin{figure}
\caption{Prediction with Expert Advice (Logarithmic Loss)}\label{fig:pwea.protocl}
\begin{algorithmic}
\FOR{each round $t=1,2,\ldots$}
\STATE Each expert $\xi \in \Xi$ issues prediction $P_{\xi,t}$
\STATE We produce prediction $Q_t$
\STATE Outcome $x_t \in \xspace$ is revealed
\STATE We incur log loss $- \ln Q_t(x_t)$.
\ENDFOR
\end{algorithmic}
\end{figure}

Algorithms for combining prediction strategies can be found in the
literature under various headings. On the one hand there are results
in (Bayesian) statistics and source coding. Most relevant is Volf and
Willems' algorithm for combining two sequential data compression
algorithms \cite{volfwillems1998}, called the ``Switching Method'',
which we discuss in Section~\ref{sec:learnrate}. Important precursors
of this work include \cite{willems96:coding_piecew,WillemsKrom1997};
the algorithms described there do not combine expert predictions but
can be used for that purpose (see Section~\ref{sec:interpolation} for
details). The tradeoff between time complexity and regret has received
substantial further analysis, see~\cite{ShamirMerhav1999,
  ShamirCostello2000, GyorgyLinderLugosi12008,
  DeRooijVanErven09,ErvenGrunwaldDeRooij2012}, but such work is
outside the scope of this introduction. On the other hand, the
learning theory community has produced a lot of work on universal
prediction under the heading ``prediction with expert advice''
\cite{cesa-bianchi2006, LittlestoneWarmuth1989, vovk1990, Vovk1999,
  HerbsterWarmuth1998, haussler1998}. In this case the experts'
predictions are not necessarily probabilistic, and scored using an
arbitrary loss function. In this paper we focus on results for
logarithmic loss, although our results apply to any mixable loss
function, as discussed in~\secref{sec:loss.functions}.

In order to give a clear, consistent and comprehensive
introduction/overview of the topic, we use the Bayesian framework to
describe the considered codes, and show how algorithms for prediction
with expert advice can often be described using a prior distribution
on \emph{sequences} of experts. This approach allows unified
description of many of the mentioned results as well as some new
models that represent interesting trade-offs between time complexity
and modelling power. Following in the footsteps of, e.g.,
\cite{willems96:coding_piecew, volfwillems1998, ShamirMerhav1999,
  Jaakkola2003}, we use hidden Markov models (HMMs) as an intuitive
graphical language to describe such priors, and obtain a
computationally efficient implementation using standard algorithms
such as the forward algorithm. Our main focus will be on algorithms
for expert tracking.

Let us emphasize that, although technically the algorithms we consider
are Bayesian, we use a worst-case individual sequence analysis. Thus,
we do not adopt the usual subjective Bayesian interpretation of the
prior distribution as an expression of belief; in fact, we do not make
any statements or assumptions about the data-generating machinery
whatsoever.

\subsection{Our Contribution}
The aim of this work is to provide a readily useful and accessible
introduction to prediction with expert advice for an information
theoretic audience. To this end we present models in terms of HMM state
transition diagrams which, although widely used, have not been applied
as consistently to the problem of prediction with expert advice. This
graphical language allows us to conveniently describe many existing
models and design new ones. The resulting diagrams can be
understood and compared with ease; moreover, computational efficiency
can be gleaned directly from their structure.

Beyond its tutorial nature, the paper also has two theoretical
contributions: first, we describe an interpolation construction that
is useful in the development and analysis of new algorithms
(Section~\ref{sec:interpolation}). Second, Lemma~\ref{lem:parambound}
provides a substantial generalisation of earlier methods from \cite{Monteleoni2003,Jaakkola2003} for analysing
the regret of expert models; the lemma is later used in the proof of
Theorem~\ref{thm:kernel_ml_regret}, which cannot be proven using
earlier results.

We describe a number of straightforward applications of the theory
developed in the paper, which can be interpreted both as examples and
as models of practical utility. A number of these models are new, in
particular the Quickly Decreasing Switching Rate model in
Section~\ref{sec:dsr} and the Ordered Experts model in Section~\ref{sec:oe}.

\subsection{Overview}
In \secref{sec:es.priors} we describe Bayesian prediction strategies
based on a prior distribution on sequences of experts, and discuss how
the resulting prediction strategies can be cast in the form of a
hidden Markov model.

The performance of the algorithms we are interested in is expressed in
terms of guarantees about their individual sequence
regret. \secref{sec:lossbounds} provides the main theoretical tools
for analysing the regret.

In \secref{sec:switching} we illustrate our approach by rendering
various models for tracking the best expert in HMM form, thus showing
the relationships between them. The seminal \emph{Fixed Share}
algorithm \cite{HerbsterWarmuth1995, HerbsterWarmuth1998} serves as a
starting point. Fixed Share has two drawbacks: first, one has to
specify a fixed \emph{switching rate} in advance; choosing a
suboptimal value here produces a linear penalty in the regret. Second,
the incurred regret depends on the number of observations $t$, even
when the optimal number of switches is bounded. To address these
problems, the switching probabilities need to be modelled differently.
\secref{sec:interpolation} explains how the part of the model that
describes switching probabilities can be isolated from the rest; we
then proceed to describe several alternative models for the switching
probabilities, and discuss how these modifications affect the regret
bound. In particular, \secref{sec:dsr} describes a new, simple and
effective approach to solve both problems associated with Fixed Share,
\secref{sec:learnrate} describes several methods for learning the switching rate from
data, and~\secref{sec:rl} describes another model that is
especially well suited to the scenario where changes in predictive
performance of experts are expected to appear in clusters.

So far, none of the considered models for expert tracking made any
assumptions as to the inner workings of, or the relationships between,
the various experts -- they are black boxes. However, as an
interesting special case we consider the scenario where the experts
are ordered.  For example, if the experts are prediction strategies
associated with a parametric model, instantiated with various
parameter values, then switches between two experts seem intuitively
more likely if they represent parameter values that are close. This
scenario is explored in~\secref{sec:oe}; the notion is taken to its
extreme in~\secref{sec:drift}, where the regret is no longer analysed
in terms of all-or-nothing ``switches'', but rather in terms of a more
smooth characterisation of the amount of ``parameter drift''.

A number of loose ends are discussed in~\secref{sec:loose.ends}. We
specifically discuss a number of useful results from the literature
that would distract from the exposition but are nevertheless too
important to skip over altogether (\ref{sec:conditioning.on.data}), as well as a variant evaluation criterion called adaptive regret  (\ref{sec:adaptive}). We then consider how one might estimate which expert made the best
prediction at a certain time (\ref{sec:map}). Finally we describe
how our results can be generalised to any mixable loss function
(\ref{sec:loss.functions}) and to online investment (\ref{sec:investment}).

\section{Expert Sequence Priors and Hidden Markov Models}\label{sec:es.priors}

Since we do not know who among our set $\Xi$ of experts will issue the
best predictions and achieve minimal codelength, the straightworward
Bayesian response is to define a prior $w$ on $\Xi$, and then
construct a distribution $\ehmm{b}[w]$ for the joint space $\xspace^\infty \times\Xi$
with
\begin{equation}\label{eq:bayes}
\ehmm{b}[w](x^t,\xi)\df w(\xi)\ehmm{P}[\xi](x^t),\text{~where~}\ehmm{P}[\xi](x^t)\df\prod_{i=1}^t P_{\xi,i}(x_i).
\end{equation}
The Bayesian prediction is obtained by conditioning on past
observations and marginalising over $\Xi$ as follows:
\[
Q_{t+1}(\rv x_{t+1}) \df \ehmm{b}[w](\rv x_{t+1}\mid x^t).
\]
(Random variables are denoted in bold face.)
Note that this prediction depends on the expert predictions from times $1$ through $t+1$ but not beyond.
This strategy is simple and effective. Compared to the single best
expert, the regret of this strategy is
\begin{equation}\label{eq:simple_regret}
-\ln \ehmm{b}[w](x^t)-\left(-\ln\ehmm{P}[\hat\xi](x^t)\right),
\end{equation}
where $\hat\xi=\argmax_\xi\ehmm{P}[\xi](x^t)$, breaking ties arbitrarily. To
bound the regret, note that
\begin{equation}\label{eq:drop_terms}
\ehmm{P}[\hat\xi](x^t)~\ge~\underbrace{\sum_{\xi\in\Xi}w(\xi)\ehmm{P}[\xi](x^t)}_{=\ehmm{b}[w](x^t)}~\ge~ w(\hat\xi)\ehmm{P}[\hat\xi](x^t);
\end{equation}
substitution in~\eqref{eq:simple_regret} reveals that the regret must
be in the interval $[0,-\ln w(\hat\xi)]$.

Thus, the good news is that this first strategy guarantees a
codelength that is within a constant of the performance of the best
available expert $\hat\xi$. On the other hand, with this strategy we
never do any \emph{better} than $\hat\xi$ either! Consider that the
nature of the data generating process may evolve over time;
consequently different experts may be better during different periods
of time. It is also possible that not the data generating process, but
the experts themselves change as more and more outcomes are being
observed: they may learn from past mistakes, possibly at different
rates, or they may have occasional bad days, etc. By generalising the
Bayesian modelling a little bit, we can compete with prediction
strategies that perform far better than the best individual expert.

\subsection{Including Transient Behaviour}
Rather than using a prior distribution to represent our uncertainty
about which \emph{single} expert is best, we generalise the setup by
considering which expert is best \emph{in each round}. Let $\pi$ be a
prior on infinite sequences of experts, called an \emph{expert
  sequence prior} (ES-prior). We subsequently define the Bayesian
joint distribution $\ehmm{B}[\pi]$ by
\begin{equation}\label{eq:ES.joint}
\ehmm{B}[\pi](x^t,\xi^t)\df\pi(\xi^t)\ehmm{P}[\xi^t](x^t),\text{~where~}\ehmm{P}[\xi^t](x^t)\df\prod_{i=1}^t P_{\xi_i,i}(x_i).
\end{equation}
We can recover the previous prediction strategy by defining $\pi$ such
that it assigns probability zero to any expert sequence that lists
more than one expert.

In the simplest case, the prior models the sequence of experts as a
Markov chain, but it is often desirable to carry along some additional
state information besides the identity of the previous expert in the
definition of the prior. Therefore, we will construct the Bayesian
prediction strategy as a Hidden Markov Model of the form depicted
in~\figref{fig:hmm}, where the state of the prior process is captured
by the $\rv \wstate_t$ variables.

Being a black box, an expert $\xi$ may use \emph{any} strategy to form
the prediction $P_{\xi,t}$ for $\rv x_t$. The framework could even
incorporate psychic experts who have some metaphysical access to
future data! Or, as is the standard assumption for proving regret
bounds, we may consider experts that conspire to maximally frustrate
the prediction task. As such, an application may involve complicated
dependencies between the $\rv x_t$ and the predictions of the experts
in general. However, as not only the data but also the predictions of
the experts are \emph{observed}, there is no need to include these
dependencies in the model.

%
%
%

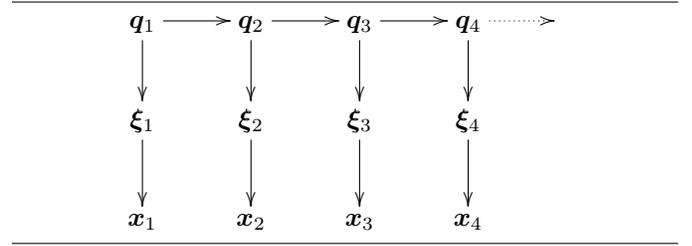
\begin{figure}
\centering
$\xymatrix{
{\rv \wstate_1} \ar[r] \ar[d] & {\rv \wstate_2} \ar[r]\ar[d] & {\rv \wstate_3}
\ar[r]\ar[d] & {\rv \wstate_4} \ar@{.>}[r]\ar[d] &
\\
{\rv \xi_1} \ar[d] & {\rv \xi_2} \ar[d] & {\rv \xi_3} \ar[d] & {\rv
  \xi_4} \ar[d] &
\\
{\rv x_1}  & {\rv x_2}  & {\rv x_3}  & {\rv x_4} &
}$
\caption{HMM description of a prediction strategy\label{fig:hmm}}
\end{figure}

\subsection{Graphical Specification of Expert HMMs}
We now outline a graphical language for describing HMMs of the form
shown in \figref{fig:hmm}, which allows us to cleanly and intuitively
display model structure. The resulting prediction strategy can be
read off directly from these diagrams, and diagrammatic simplicity
implies computational efficiency.

An Expert Hidden Markov Model (EHMM) is a joint distribution on
sequences of states, experts and outcomes. It is defined by the
following ingredients. We start by choosing a set of \markdef{states}
$\wstates$ and a designated \markdef{start state}
$\wstart\in\wstates$. We then specify the \markdef{transition
  probability} between states using a Markov kernel $\wtf$. (We now
have a regular Markov chain on the sequence of states.)  A subset
$\wstates^\prd \subseteq \wstates$ of the states are called
\markdef{productive}. Experts are deterministically assigned to the
productive states by $\wnl:\wstates^\prd\to\Xi$.  The role of the
non-productive \markdef{silent} states is to provide fine-grained
independencies in the Markov chain, allowing us to ``spell out'' the
transitions between productive states conceptually and computationally
efficiently.

Given expert predictions $P_{\xi,i}$ for all $\xi\in\Xi$ and $i = 1, 2, \ldots$, the joint probability is given as follows. Let $\wstate^\lambda$ be a sequence of states, and let $\wstate^\prd_1, \ldots, \wstate^\prd_t$ be the subsequence of its productive states. Then
\begin{multline}\label{eq:EHMM.joint}
\Q(\wstate^\lambda, \xi^t, x^t) 
~\df~
\\
\begin{cases}
\displaystyle
\ehmm{P}[\xi^t](x^t)
\prod_{i=0}^{\lambda-1} \wtf(\wstate_i \to \wstate_{i+1})
&\text{if $\forall i:\wnl(\wstate^\prd_i)=\xi_i$,}\\
0&\text{otherwise.}
\end{cases}
\end{multline}
For convenience, we identify the EHMM distribution with its defining
5-tuple,
i.e. $\Q\equiv\tuple{\wstates,\wstates^\prd,\wstart,\wtf,\wnl}$.

The advantage of our setup is that we can specify EHMMs using
intuitive state transition diagrams, as done for example in Figures~\ref{graph:bayes}
and~\ref{graph:fixmix}. We first draw a node $N_\wstate$ for each
state $\wstate$. We use an open dot $\xymatrix{\bn}$ for the start
state, black dots $\xymatrix{\si}$ for the silent states, and we
display each productive state $\wstate$ as an open circle labelled by
the expert $\wnl(\wstate)$ who is assigned to make the prediction,
like this: $\xymatrix{\ex{a}}$. We draw an arrow from $N_\wstate$ to
$N_{\wstate'}$ if the transition probability $\wtf(\wstate \to
\wstate')$ is nonzero. The transition probabilities themselves could
be written along such arrows, but this quickly becomes messy. Instead
we write them below our graphs.

The forward algorithm, see e.g.\ \cite{rabiner1989}, can be used to compute the predictive distribution on experts given past data, i.e.\
\[
\Q\delc*{\rv \xi_{t+1}}{x^t}
.
\]
Intuitively, this is done by maintaining weights on the productive states. These weights are then alternately used for prediction and subsequently conditioned on observations, and percolated forward through the network of silent states according to the transition probabilities $\wtf$.

The total running time of the forward algorithm for all time steps is proportional to the number of
edges in the graph.

\subsection{Examples}
We give the ES-priors and EHMMs that correspond to the simplest models:
Bayesian mixtures and elementwise mixtures with fixed parameters.

\begin{figure}
\caption{Bayesian mixture $\ehmm{b}[w]$}\label{graph:bayes}
\centering
\subfloat{%
$\xymatrix@R=0.5em@C=1.7em{
&\ex{a}\name{\expname a,1} \ar[rr] && \ex{a} \name{\expname a,2} \ar[rr] && \ex{a} \ar[rr] && \ex{a} \ar@{.>}[r] & \\
\\
&\ex{b} \name{\expname b,1} \ar[rr] && \ex{b} \name{\expname b,2} \ar[rr] && \ex{b} \ar[rr] && \ex{b} \ar@{.>}[r] & \\
\bn \name[+<0pt,10pt>]{0}  \ar[ur] \ar[uuur] \ar[dr] \ar[dddr]\\
&\ex{c} \name{\expname c,1} \ar[rr] && \ex{c} \name{\expname c,2} \ar[rr] && \ex{c} \ar[rr] && \ex{c} \ar@{.>}[r] & \\
\\
&\ex{d} \name[+<4pt,10pt>]{\expname d,1} \ar[rr] && \ex{d} \name{\expname d,2}\ar[rr] && \ex{d} \ar[rr] && \ex{d} \ar@{.>}[r] & \\
}$}
\quad
\subfloat{%
\setlength{\abovedisplayskip}{0cm}
\setlength{\belowdisplayskip}{0cm}
\begin{minipage}[t]{.4\textwidth}
\begin{gather*}
\ehmm{b}[w] =
  \tuple{\wstates, \wstates^\prd, 0, \wtf, \wnl}
\\
\wstates= \wstates^\prd \cup \set{0} 
\quad \wstates^\prd =  \Xi \times \posints 
\quad \wnl(\xi,t) = \xi
\\
\wtf\del*{
\begin{aligned}
\tuple{0} &\to \tuple{\xi, 1}
\\
\tuple{\xi, t} &\to \tuple{\xi, t+1}
\end{aligned}} = \del*{
\begin{gathered}
\mw(\xi)
\\
1
\end{gathered}}
\end{gather*}
\end{minipage}
}
\end{figure}
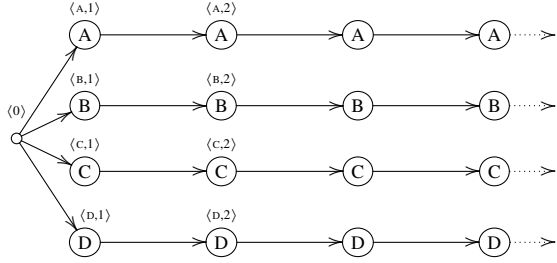

\begin{example}[Bayesian Mixtures]\label{example:bayes.es.prior}\label{example:bayes.hmm}
  The EHMM $\ehmm{b}[w]$ for the Bayesian mixture is shown in
  \figref{graph:bayes}. 
The figure illustrates that in the standard Bayesian mixture there is no provision for a-priori redistribution of probability mass between experts: the posterior distribution will concentrate very quickly on the state corresponding to the expert that assigns the overall highest likelihood to the data.
The vector-style definition of $\wtf$ is a
  shorthand; each row specifies one or more transition
  probabilities. The reader may check that this EHMM formally
  corresponds to the Bayesian mixture in the sense that the marginal
  likelihoods on data of \eqref{eq:EHMM.joint} and \eqref{eq:bayes}
  coincide.
  It is well-known that the Bayesian prediction
  $\ehmm{b}[w]\delc{x_{t+1}}{x^t}$ can be computed in $O(k)$ time per trial by
  maintaining the posterior. The forward algorithm on $\ehmm{b}[w]$
  has the same efficiency.
\qedex
\end{example}

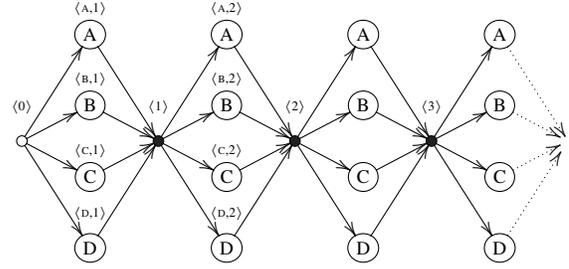
\begin{figure}
\caption{Fixed elementwise mixture $\ehmm{em}[\mw]$}\label{graph:fixmix}
\centering
\subfloat{%
$\xymatrix@R=0.5em@C=1.7em{
&\ex{a} \name{\expname a,1} \ar[rddd] && \ex{a} \name{\expname a,2} \ar[rddd] && \ex{a} \ar[rddd] && \ex{a} \ar@{.>}[rddd] & \\
\\
&\ex{b} \name{\expname b,1} \ar[rd] && \ex{b} \name{\expname b,2} \ar[rd] && \ex{b} \ar[rd] && \ex{b} \ar@{.>}[rd] & \\
\bn \name[+<0pt,13pt>]{0} \ar[ur] \ar[uuur] \ar[dr] \ar[dddr]&&
\si \name[+<0pt,13pt>]{1} \ar[ur] \ar[uuur] \ar[dr] \ar[dddr]&&
\si \name[+<0pt,13pt>]{2} \ar[ur] \ar[uuur] \ar[dr] \ar[dddr]&&
\si \name[+<0pt,13pt>]{3} \ar[ur] \ar[uuur] \ar[dr] \ar[dddr]&&
\\
&\ex{c} \name{\expname c,1} \ar[ru] && \ex{c} \name{\expname c,2} \ar[ru] && \ex{c} \ar[ru] && \ex{c} \ar@{.>}[ru] & \\
\\
&\ex{d} \name[+<0pt,13pt>]{\expname d,1} \ar[ruuu] && \ex{d} \name[+<0pt,13pt>]{\expname d,2} \ar[ruuu] && \ex{d} \ar[ruuu] && \ex{d} \ar@{.>}[ruuu] &
}$
}
\quad
\subfloat{%
\setlength{\abovedisplayskip}{0cm}
\setlength{\belowdisplayskip}{0cm}
\begin{minipage}[t]{.4\textwidth}
\begin{gather*}
\ehmm{em}[\mw] ~=~  \tuple{\wstates, \wstates^\prd, 0, \wtf, \wnl}
\\
\wstates = \wstates^\sil \cup \wstates^\prd
\quad \wstates^\sil =\nats
\quad \wstates^\prd = \Xi \times \posints
\quad \wnl(\xi, t) = \xi
\\
\wtf\del*{
\begin{aligned}
\tuple{t} &\to \tuple{\xi, t+1}
\\
\tuple{\xi,t} &\to \tuple{t}
\end{aligned}} = \del*{
\begin{gathered}
\mw(\xi)
\\
1
\end{gathered}}
\end{gather*}
\end{minipage}
}
\end{figure}%

\begin{example}[Elementwise Mixtures]\label{example:fixed.elementwise.mixtures}\label{example:fixed.elementwise.mixtures.hmm}
  The \markdef{elementwise mixture}\footnote{These mixtures are
    sometimes just called mixtures, or predictive mixtures, or exponentially weighted averages. We use the
    term elementwise mixtures to avoid confusion with Bayesian mixtures.} with 
  mixture weights $\mw$ on experts $\Xi$ predicts as follows:
\[
Q_t(x_t) ~\df~ \sum_{\xi \in \Xi} P_{\xi,t}(x_t) \mw(\xi).
\]
This prediction strategy for elementwise mixtures can be implemented
by the EHMM $\ehmm{em}[\mw]$ defined in \figref{graph:fixmix}.
The EHMM has a single silent state per outcome, whose transition
probabilities are the mixture weights $\mw$.
Intuitively, by funnelling all weight through a single silent state all differentiation between the experts based on past performance is forgotten. As such no learning occurs, the strength of this model instead lies in the fact that the model always assigns reasonable probability to any outcome that is likely according to any of the experts.
The forward algorithm calculates the prediction of $\ehmm{em}[\mw]$ in
$O(k)$ time per trial.  \qedex
\end{example}

\section{Regret Bounds}\label{sec:lossbounds}
Here we provide some handles for analysing the predictive performance
of EHMMs.  In each case, the idea is to compare the loss incurred by
some model $\Q$ to the loss incurred by another prediction strategy
from a set $\mathcal M$ of reference strategies. For example,
$\mathcal M$ might be the set of all prediction strategies based on a
fixed expert sequence that contains $m$ blocks, or it might be the set
of all prediction strategies that can be obtained by mixing over the
predictions of the experts with fixed weights. The goal is now to
provide an upper bound on the excess codelength of the model $\Q$
compared to the best of these reference strategies. Note that for many
models we provide \emph{simultaneous} regret guarantees with respect
to several distinct reference classes; e.g. the regret bounds for
expert tracking algorithms hold for all numbers of blocks $m$.

Throughout this paper, we use three types of regret bounds, which are
given in order of increasing sophistication. The first bound is
appropriate if only a few expert sequences contribute significantly to
the probability of the data. In that case it is sufficient to simply
drop some terms from the Bayesian mixture.

\begin{lemma}[Regret w.r.t.\ Expert Sequence $\xi^t$]\label{lem:dropmarg}
  Let $\Q$ denote an EHMM, and let $\xi^t$ denote a particular
  reference expert sequence. Then, for all data $x^t$,
\begin{equation}\label{eq:bound1}
\ln \frac{\ehmm{P}[\xi^t](x^t)}{\Q(x^t)}\le-\ln\Q(\xi^t).
\end{equation}
\end{lemma}
\begin{IEEEproof}The bound is obtained by dropping all terms in the
  mixture except for the one corresponding to $\xi^t$.
\end{IEEEproof}

We obtain an expression for the regret w.r.t.\ some reference set
$\mathcal M \subseteq \Xi^t$ by maximising $\xi^t$ over its elements. We have
already seen an example application: the regret
bound~\eqref{eq:drop_terms} for $\Q=\ehmm{b}[w]$ is derived in
this way.

In the second kind of bound, another EHMM $\R$ plays the role
of reference prediction strategy. It can be useful even if the number
of different expert sequences with significant contribution to the
probability is very large. The following lemma forms the basis for
such bounds.

\newcommand{\V}{{\ehmm{V}[]}}
\begin{lemma}[Regret w.r.t.\ another EHMM]\label{lem:toexpectation}
  Fix data $x^t$ and EHMMs $\Q$ and $\R$. We have
  \[
\ln \frac{\R(x^t)}{\Q(x^t)}
~\le~
-\ln \Exp_\V \sbr[\bigg]{\frac{\Q(\rv{\xi}^t)}{\R(\rv{\xi}^t)}}
~\le~
\Exp_\V \sbr[\bigg]{\ln \frac{\R(\rv{\xi}^t)}{\Q(\rv{\xi}^t)}}
,\]
where $\V(\rv\xi^t)=\R(\rv\xi^t|x^t)$.
\end{lemma}
\begin{IEEEproof}
  Rewrite
  \[\begin{split}
  \frac{\Q(x^t)}{\R(x^t)}
&\ge~\kern-1.5em\sum_{\xi^t : \R(x^t,\xi^t) > 0}\kern-0.7em\frac{\R(x^t,\xi^t)}{\R(x^t)}\cdot\frac{\Q(x^t,\xi^t)}{\R(x^t,\xi^t)}\\
&=~\Exp_\V\sbr*{\frac{\Q(x^t,\rv{\xi}^t)}{\R(x^t,\rv{\xi}^t)}}
~=~\Exp_\V\sbr*{\frac{\Q(\rv{\xi}^t)}{\R(\rv{\xi}^t)}},\end{split}\]
  take the $-\ln$ and subsequently apply Jensen's inequality.
\end{IEEEproof}
%
\newcommand{\Qpara}{\wstates^\dagger}
\newcommand{\cnt}{n}
\newcommand{\tsp}{\top}

\noindent Although this bound still involves the actual data through
the distribution $\V$, sometimes the expectation can be replaced by a
maximum over $\xi^t$. This may be sufficiently sharp for the job at
hand if a good uniform bound is available for the ES-priors.  However,
it is possible to say more about the regret if $\Q$ and $\R$ share a certain
structure.


The next and final lemma applies to EHMMs in which some of the transition probabilities are a function of a parameter vector $\beta$. Intuitively, models with fewer parameters have more constrained transition dynamics and are hence less expressive, making it easier to compete with the maximum likelihood parameter values. 

The lemma is a generalisation of Theorems~1 and~3
in~\cite{Monteleoni2003}\footnote{Theorem~1 also appears
in~\cite{Jaakkola2003}. Although both theorems are valid, Theorem~3
does not generalise Theorem~1; there is a problem with Lemma~3.3.1 in
the cited work.}. Theorem~1 is concerned with the special case of Fixed Share, where the transition matrix is parameterised by the switching rate $\alpha$. Theorem~3 applies to unconstrained Markov transition dynamics (with $\card{\Xi}^2$ parameters). 
The lemma below yields sharper results for restricted transition dynamics that can be expressed in exponential family form. 
%
We will apply this result to reobtain Monteleoni \& Jaakkola's sophisticated regret bound \cite{Jaakkola2003} of Fixed Share (in Theorem~\ref{thm:fs.parambound}) and then use it for the kernel models of Sections~\ref{sec:oe} and~\ref{sec:drift} (Theorem~\ref{thm:kernel_ml_regret}), for which the added generality is essential.

The parameterisation of the EHMM should be of the following form.  Let
$T_\beta(j)=e^{\beta^\tsp\phi(j)}h(j)/Z(\beta)$ be an exponential
family of distributions with parameter vector $\beta$, some sufficient
statistic $\phi$, carrier $h$ and normalisation $Z(\beta)=\sum_j
e^{\beta^\tsp\phi(j)}h(j)$, where $j$ takes values in a finite set
$\mathcal J$.  Let $\Qpara\subseteq\wstates$ be a subset of the state
space for which the transition probabilities are governed by this
exponential family model $T_\beta$ in the sense that there is an
injective function $S:\Qpara\times\mathcal J\to\wstates$ that
indicates for each such state how the symbols in $\mathcal J$ map to
the successor states:
\[
\wtf(\wstate\to S(q,j))=T_\beta(j).
\]
Intuitively, the more often such parameterised transitions are traversed, the larger the difference between two distinct values for $\beta$.
Therefore we need to count the number of transitions from states in $\Qpara$. 
Define the random variable $\rv \wstate^{(t)}(\wstate^\infty) \df \wstate^\lambda$ where $\lambda$ is chosen such that $\wstate_\lambda$ is the $t$th productive state of $\wstate^\infty$. We assume throughout that $\rv \wstate^{(t)}$ is well-defined almost surely for all $t$. Further define $\cnt_j(\wstate^{(t)}) \df \card{\setc{i}{0\le i<\lambda,\wstate_i\in
    \Qpara,S(\wstate_i,j)=\wstate_{i+1}}}$, where $\lambda$ is the length of $\wstate^{(t)}$, so that $\cnt_j(\wstate^{(t)})$ is the number of
transitions in $\wstate^{(t)}$ between states in $\Qpara$ and their
$j$-successors, and let $\cnt(\wstate^{(t)})=\sum_j
\cnt_j(\wstate^{(t)})$ be the total parametrised transition count. We can now state our result:

\newcommand{\W}{{\ehmm{W}[]}}
\begin{lemma}[Regret w.r.t.\ ML Parameter $\hat\beta$]\label{lem:parambound}
  Let $\ehmm{q}[\beta]$ be parameterised as defined above. Fix
  outcomes $x^t$ and let $\hat\beta=\argmax_\beta \ehmm{q}[\beta](x^t)$ with the assumption that $\ehmm{q}[\hat\beta](x^t)>0$.
  Furthermore let $\W(\rv \wstate^{(t)})=\ehmm{Q}[\hat\beta](\rv\wstate^{(t)}|x^t)$ denote the posterior distribution of $\rv \wstate^{(t)}$ under prior $\ehmm{q}[\hat\beta]$. We then have
  \[\begin{split}
  \ln\frac{\ehmm{q}[\hat\beta](x^t)}{\ehmm{q}[\beta](x^t)}
  &~\le~
  -\ln \Exp_\W\sbr*{\frac{\ehmm{q}[\beta](\rv{\wstate}^{(t)})}{\ehmm{q}[\hat\beta](\rv{\wstate}^{(t)})}}
  ~\le~
  \Exp_\W\sbr*{\ln\frac{\ehmm{q}[\hat\beta](\rv{\wstate}^{(t)})}{\ehmm{q}[\beta](\rv{\wstate}^{(t)})}}\\
  &~=~
  \Exp_\W[\cnt(\rv{\wstate}^{(t)})]\kld\delcc[\big]{T_{\hat\beta}}{T_\beta},
  \end{split}\]
where $\kld$ is the Kullback-Leibler divergence.
\end{lemma}
\noindent As before, the goal of this lemma is to say something about
the overhead incurred by using a particular strategy $\Q_\beta$ instead
of the reference strategy $\Q_{\hat\beta}$. The result still depends on
the data via the distribution $\W$, but in applications of the lemma
the idea will be to replace $\Exp_\W[\cnt(\rv{\wstate}^{(t)})]$ by a
bound on the number of states in $\Qpara$ that may be traversed by relevant
state sequences.

As an important special case, the distribution on $\mathcal J$ may be
fully specified by a multinomial distribution with parameter vector
$\mw$; we can then apply the lemma above by setting
$\phi(j)=(0,\ldots,1,0,\ldots)$, with the $1$ appearing at the
$j^\text{th}$ position, and $h(j)=1$. Instantiated in this way the
Lemma expresses a bound with respect to all possible distributions on
the successor state, much like Theorem~3 of
\cite{Monteleoni2003}. However, if we choose a more constrained
exponential family, the reference strategy $T_{\hat\beta}$ has fewer
degrees of freedom and the divergence
$\kld\delcc{T_{\hat\beta}}{T_\beta}$ that appears in the bound is
reduced. This can be used to obtain
Theorem~1 of \cite{Monteleoni2003} (which also appears here as
Theorem~\ref{thm:fs.parambound}); we also use it to prove
Theorem~\ref{thm:kernel_ml_regret} below.
\begin{IEEEproof}[Proof of \lemref{lem:parambound}]
The first two inequalities are \lemref{lem:toexpectation} on the level of state sequences. The
  contribution of this lemma lies in the last equality. 
First expand

  \begin{equation}\label{eq:expln}
    \begin{split}
&\Exp_\W\sbr*{\ln\frac{\ehmm{q}[\hat\beta](\rv{\wstate}^{(t)})}{\ehmm{q}[\beta](\rv{\wstate}^{(t)})}}
~=~\Exp_\W\sbr*{\sum_{j\in\mathcal
    J}\cnt_j(\rv{\wstate}^{(t)})\ln\frac{T_{\hat\beta}(j)}{T_{\beta}(j)}}
\\
&~=~\sum_j
\Exp_\W[\cnt_j(\rv{\wstate}^{(t)})]\del*{(\hat\beta-\beta)^\tsp\phi(j)+\ln\frac{Z(\beta)}{Z(\hat\beta)}}
\\
&~=~(\hat\beta-\beta)^\tsp\sum_j\phi(j)\Exp_\W[\cnt_j(\rv{\wstate}^{(t)})]+\Exp_\W[\cnt(\rv{\wstate}^{(t)})]\ln\frac{Z(\beta)}{Z(\hat\beta)}.
\end{split}
\end{equation}
Since $\hat\beta$ maximises the probability of $\ehmm{q}[\beta]$ and
since $\nabla_\beta \ln \ehmm{q}[\beta](x^t,\wstate^{(t)}) = \nabla_\beta \ln
\del[\big]{\ehmm{q}[\beta](\wstate^{(t)})/\ehmm{q}[\hat\beta](\wstate^{(t)})}$
we obtain\footnote{ Recall that for a function $f : \reals^k \to
  \reals$, the vector differential $\nabla_\beta f$ is defined as
  the column vector $\del*{\frac{\partial f}{\partial \beta_1},
    \ldots, \frac{\partial f}{\partial \beta_k} }$.  }
\newcommand{\delval}[2][\bigg]{\nabla_\beta #2 #1|_{\beta=\hat\beta}}
\[\begin{split}
\vec 0 
&=
-\delval{ \frac{\ehmm{q}[\beta](x^t)}{\ehmm{q}[\hat\beta](x^t)} }=
-\frac{\sum_{\wstate^{(t)}} \delval[\big]{ \ehmm{q}[\beta](x^t,\wstate^{(t)})}}{\ehmm{q}[\hat\beta](x^t)}\\
&=
-\sum_{\wstate^{(t)}} \frac{\ehmm{q}[\hat\beta](x^t,\wstate^{(t)})}{\ehmm{q}[\hat\beta](x^t)} \delval[\big]{ \ln \ehmm{q}[\beta](x^t,\wstate^{(t)})}\\
&=
\delval[\bigg]{ \Exp_\W \sbr[\bigg]{  \ln \frac{\ehmm{q}[\hat\beta](\rv\wstate^{(t)})}{\ehmm{q}[\beta](\rv\wstate^{(t)})}}}.
\end{split}\]
This shows that the vector differential of~\eqref{eq:expln} must be zero at $\hat \beta$.
Reordering terms we obtain
\begin{equation}\label{eq:deriv}\begin{split}
\sum_j\phi(j)\Exp_\W[\cnt_j(\rv{\wstate}^{(t)})]
&=
\Exp_\W[\cnt(\rv{\wstate}^{(t)})]\nabla_\beta\ln\frac{Z(\beta)}{Z(\hat\beta)}\bigg|_{\beta=\hat\beta}\\
&=
\Exp_\W[\cnt(\rv{\wstate}^{(t)})]\mathop{\Exp}_{\rv{j}\sim T_{\hat\beta}}[\phi(\rv{j})],
\end{split}\end{equation}
where the last step follows from
\[
  \nabla_\beta\ln {Z(\beta)}
=\frac{\nabla_\beta
    Z(\beta)}{Z(\beta)}
=\sum_{j\in\mathcal
    J}\frac{e^{\beta^\tsp\phi(j)}h(j)}{Z(\beta)}\phi(j)
=\mathop{\Exp}_{\rv{j}\sim
    T_{\beta}}[\phi(\rv{j})].
\]
Using~\eqref{eq:deriv} we may now simplify~\eqref{eq:expln} to
\[\begin{split}
  &\Exp_\W\sbr*{\ln\frac{\ehmm{q}[\hat\beta](\rv{\wstate}^{(t)})}{\ehmm{q}[\beta](\rv{\wstate}^{(t)})}}\\
&=
\Exp_\W[\cnt(\rv{\wstate}^{(t)})]\del*{(\hat\beta-\beta)^\tsp\mathop{\Exp}_{\rv{j}\sim
      T_{\hat\beta}}[\phi(\rv{j})]+\ln\frac{Z(\beta)}{Z(\hat\beta)}}\\
&=\Exp_\W[\cnt(\rv{\wstate}^{(t)})]\kld\delcc[\big]{T_{\hat\beta}}{T_\beta},
\end{split}\]
completing the proof.
\end{IEEEproof}
We now apply the lemma to our two running examples. Note that in
both cases, the model is parameterised such that the total number of
parameterised transitions is known.

\begin{example}[Bayesian Mixture Regret]\label{ex:bayesbound2}
  We have already obtained the bound~\eqref{eq:drop_terms} for
  $\ehmm{b}[w]$ using \lemref{lem:dropmarg}, but it is
  instructive to do the same using \lemref{lem:parambound}.  Let
  $\Qpara$ contain just the initial silent state and identify the experts
  $\Xi$ with $\mathcal J$. We now have $\Exp_\W[\cnt(\rv{\wstate}^{(t)})]=1$, so the
  lemma tells us that our regret is $\kld\delcc{\hat w}{w}$,
  where $\hat w$ is the \emph{hindsight optimal} prior weight
  vector that maximises the probability of the available data, and
  $w$ is the prior we actually use. Now observe that in order to
  maximise probability, $\hat w$ must assign all mass to a single
  expert $\hat\xi$, so $\kld\delcc{\hat w}{w}=-\ln w(\hat\xi)$ as
\  before.
\qedex
\end{example}

\begin{example}[Elementwise Mixture Regret]\label{ex:elmixbound}
  We now compute the regret of $\ehmm{em}[\mw]$. Let $\Qpara$
  contain the silent states and again identify the experts $\Xi$ with
  $\mathcal J$.  For elementwise mixtures,
  $\Exp_\W[\cnt(\rv{\wstate}^{(t)})]=t$.  So by \lemref{lem:parambound},
  the regret of predicting the outcomes $x^t$ with an elementwise
  mixture with weights $\mw$ instead of the \emph{hindsight optimal}
  mixture weights $\hat\mw$ is bounded by $t
  \kld\delcc{\hat\mw}{\mw}$.  \qedex
\end{example}

\section{Switching Strategies}\label{sec:switching}

\subsection{Fixed Share}\label{sec:fs}
%

Mark Herbster and Manfred Warmuth's paper on \emph{tracking the best
  expert} \cite{HerbsterWarmuth1995, HerbsterWarmuth1998} is the first
to consider the scenario where the best predicting expert may change
over time.  They compare the loss of their algorithm to the smallest
loss that can be achieved by splitting the data of size $t$ into $m$
segments, and within each segment, copying the predictions of the
expert who in hindsight turns out to be best for that particular
segment. They give two algorithms called Fixed Share and Variable
Share, but the motivation for the second algorithm applies only to
loss functions other than log-loss (see
\secref{sec:conditioning.on.data}), so we focus on Fixed Share, which
matches the EHMM $\ehmm{fs}[w,\alpha]$ defined in
\figref{graph:mixedshare}.  Note that all arcs \emph{into} the silent
states have fixed probability $\alpha\in[0,1]$ and all arcs
\emph{from} the silent states have some fixed distribution $w$ on
$\Xi$ (the original algorithm uses a uniform $w(\xi)=1/k$). The same
algorithm is also described as an instance of the Aggregating
Algorithm in~\cite{Vovk1999}. Fixed Share reduces to fixed elementwise
mixtures by setting $\alpha=1$ and to Bayesian mixtures by setting
$\alpha=0$.
\begin{figure}
\caption{Fixed Share: $\ehmm{fs}[w,\alpha]$}\label{graph:mixedshare}
\centering
\subfloat{%
$\xymatrix@R=0.5em@C=1.6em{
&\ex{a}\ar[dddr] \ar[rr]&&\ex{a}\ar[dddr] \ar[rr]&&\ex{a}\ar[dddr] \ar[rr]&&\ex{a}\ar@{.>}[r] \ar@{.>}[dddr]&\\
&&&&&&&\\
&\ex{b}\ar[dr]\ar[rr]&&\ex{b}\ar[dr]\ar[rr]&&\ex{b}\ar[dr]\ar[rr]&&\ex{b}\ar@{.>}[r] \ar@{.>}[dr]&\\
\bn\name[+<1pt,9pt>]{0} \ar[uuur]\ar[ur]\ar[dr]\ar[dddr]&&
\si\name[+<1pt,9pt>]{1} \ar[uuur]\ar[ur]\ar[dr]\ar[dddr]&&
\si\name[+<1pt,9pt>]{2} \ar[uuur]\ar[ur]\ar[dr]\ar[dddr]&&
\si\name[+<1pt,9pt>]{3} \ar[uuur]\ar[ur]\ar[dr]\ar[dddr]&&
\\
&\ex{c}\ar[ur]\ar[rr]&&\ex{c}\ar[ur]\ar[rr]&&\ex{c}\ar[ur]\ar[rr]&&\ex{c}\ar@{.>}[r] \ar@{.>}[ur]&\\
&&&&&&&\\
&\ex{d}\ar[uuur]\ar[rr]&&\ex{d}\ar[uuur]\ar[rr]&&\ex{d}\ar[uuur]\ar[rr]&&\ex{d}\ar@{.>}[r] \ar@{.>}[uuur]&
}$
}
\quad
\subfloat{%
\setlength{\abovedisplayskip}{0cm}
\setlength{\belowdisplayskip}{0cm}
\begin{minipage}[t]{.42\textwidth}
\begin{gather*}
\begin{aligned}
\wstates &= \wstates^\sil \cup \wstates^\prd
&
\wstates^\sil &= \nats
&
\wstates^\prd &= \Xi \times \posints
\end{aligned}
\\
\begin{aligned}
\wstart &= 0
&
\wnl(\xi,t) &= \xi
\end{aligned}
\\
\wtf\del*{\begin{aligned}
\tuple{t} &\to \tuple{\xi,t+1}
\\
\tuple{\xi, t} &\to \tuple{t}
\\
\tuple{\xi, t} &\to \tuple{\xi,t+1}
\end{aligned}} 
= 
\del*{\begin{gathered}
w(\xi) 
\\
\alpha
\\
1-\alpha
\end{gathered}}
\end{gather*}
\end{minipage}
}%
\end{figure}
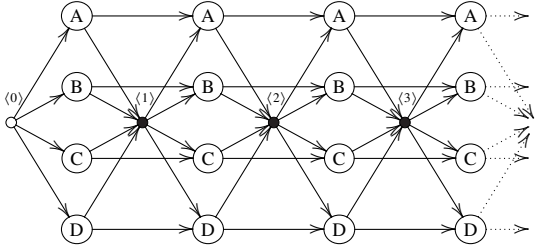%
Each productive state represents that a particular expert is used at a
certain sample size. Once a transition to a silent state is made, all
expert history is forgotten and a new expert is chosen according to
$w$.

We now bound the regret of Fixed Share with respect to a given
partition, i.e.\ sequence of experts.
\begin{theorem}[Herbster and Warmuth \cite{HerbsterWarmuth1998}]\label{thm:fixed.share.loss.bound}
  Fix experts $\Xi$ and data $x^t$, and let $\xi^t$ be a sequence of
  experts with $m$ blocks (i.e.\ $\card{\setc{1 \le i < t}{\xi_{i} \neq \xi_{i+1}}}=m-1$), $k=\card{\Xi}$, and $w(\xi) = 1/k$. Let
  $\alpha^* = (m-1)/(t-1)$ denote the switching frequency in
  $\xi^t$. Write
  $\ent(\alpha^*,\alpha)=-\alpha^*\ln\alpha-(1-\alpha^*)\ln(1-\alpha)$
  for the cross entropy. Then
\begin{equation}\label{eq:fixedshare.lossbound}
\ln \frac{ \ehmm{P}[\xi^t](x^t)}{\ehmm{fs}[w,\alpha](x^t)}  
~\le~ 
m \ln k
+ (t-1) \ent(\alpha^*,\alpha).
\end{equation}
\end{theorem}
\begin{IEEEproof}
Let $\wstate^{(t)}$ be the state sequence that produces $\xi^t$ and that passes through silent state $\tuple{i}$ iff $\xi_i \neq \xi_{i+1}$.
Then
\[\begin{split}
\frac{\ehmm{fs}[w,\alpha](x^t)}{ \ehmm{p}[\xi^t](x^t)}
&\mathop{\ge}^{\text{by~\eqref{eq:bound1}}}~
\ehmm{fs}[w,\alpha](\xi^t)
~\ge~
\ehmm{fs}[w,\alpha](\wstate^{(t)})\\
&=
k^{-m} (1-\alpha)^{t-m}\alpha^{m-1}.
\end{split}\]
Taking logarithms and substituting $\alpha^*$ completes the proof.
\end{IEEEproof}
Note that in Herbster and Warmuth's algorithm, switches to the same
expert are disallowed, allowing them to derive a bound with $\ln k +
(m-1) \ln (k-1)$ instead of our $m \ln k$. We do allow switches to the
same expert in all our models to keep the exposition clean and simple,
but this can be changed in the same way that this is done in the Fixed
Share paper.

While $\alpha^*$ optimises the bound, it does not necessarily maximise
the probability of the data. We therefore also calculate the regret of
$\ehmm{fs}[w,\alpha]$ with respect to the set of Fixed Share
algorithms for all $\alpha\in[0,1]$. 

\begin{theorem}[Monteleoni and Jaakkola
  \cite{Jaakkola2003}]\label{thm:fs.parambound}
For all data $x^t$ and switching rate $\alpha\in[0,1]$,
\[
\ln \frac{\ehmm{fs}[w,\hat\alpha](x^t)}{\ehmm{fs}[w,\alpha](x^t)}
~\le~ 
(t-1)\kld\delcc*{\hat\alpha}{\alpha},\]
where $\hat\alpha$ maximises the Fixed Share likelihood.
\end{theorem}
\begin{proof}
  Apply \lemref{lem:parambound}, setting $\Qpara$ to $\wstates^\prd$,
  the set of all productive states, whose outgoing transitions are
  parameterised by the switching rate $\alpha$.
\end{proof}
Judging from this theorem and from \eqref{eq:fixedshare.lossbound},
the regret appears to grow linearly with time, but if we substitute
the switching rate $\alpha=\alpha^*$ that optimises the bound, cross
entropy reduces to ordinary entropy, and we find that if $m$ is kept
fixed, the regret only has a logarithmic dependence on $t$: we have
\begin{equation}\label{eq:fs.sandwich}
  (m-1)\ln\frac{t-1}{m-1}\le(t-1)\ent(\alpha^*)\le(m-1)\ln\frac{t-1}{m-1}+m
\end{equation}
where $\ent(\alpha^*) = \ent(\alpha^*,\alpha^*)$ is the binary entropy.
The problem is that such asymptotics can only be achieved if we are
somehow able to guess the optimal switching rate before observing the
data. This issue is addressed in the following sections. We will
evaluate the performance of the other models for expert tracking using
the loss of Fixed Share with $\alpha=\alpha^*$ as a baseline.

\subsection{Intermezzo: Interpolation}\label{sec:interpolation}
%
%
Note how Fixed Share (\figref{graph:mixedshare}) interpolates between
the Bayesian Mixture (\figref{graph:bayes}) and the Elementwise
Mixture (\figref{graph:fixmix}). The parameter $\alpha$ determines
\emph{when} switches occur.  If no switch occurs then the Bayesian
Mixture's transitions are used: all experts' weights remain unchanged.
On the other hand, if a switch occurs then the Elementwise Mixture's
transitions are used: all experts' weights are gathered and
redistributed. Fixed Share can thus be interpreted as an algorithm
that interpolates between the Bayesian and Elementwise mixtures.

In this section we first give an intuitive high-level definition of
interpolation, and then follow it up with a detailed definition that
carefully manipulates silent states for the sake of efficiency.

Interpolations are natural to the switching
domain. In~\cite{DeRooijVanErven09}, Bernoulli HMMs are used to
produce switching rates; these are really just interpolators as
defined below. In similar fashion, interpolation allows us to lift
pretty much any algorithm for predicting binary data, such as
\cite{willems96:coding_piecew,WillemsKrom1997}, to the context of
prediction with expert advice.

Interpolations are useful as a tool to build models in a modular
fashion. For example, the Bayesian Mixture model
(\figref{graph:bayes}) can be interpreted as a method for learning
which expert is best; the Fixed Share algorithm augments this model by
introducing the possibility to reset the weights, so that the Bayesian
learning process starts anew. By reinterpreting Fixed Share as an
interpolation between the Bayesian Mixture and the Elementwise
Mixture, as we do below, it becomes possible to replace the Bayesian
component of Fixed Share with any other EHMM. Useful examples include
EHMMs that are designed to learn elementwise mixture
coefficients. This endows the model with the option to reset its
state, effectively restarting the learning process. Such possibilities
are explored in detail
in~\cite{koolen09:_fixed_share_for_learn_exper}.

\newcommand{\C}{{\ehmm{C}[]}}
\newcommand{\Cstates}{C}
\newcommand{\Cstate}{c}
\newcommand{\Cstart}{\Cstate_0}
\newcommand{\Cfake}{\Cstate_\star}
\newcommand{\Cnl}{\Sigma}
\newcommand{\Ctf}{\wtf} 

\newcommand{\Rstates}{R}
\newcommand{\Rstate}{r}
\newcommand{\Rstart}{\Rstate_0}
\newcommand{\Rnl}{\wnl^\otimes} 
\newcommand{\Rtf}{\wtf^\otimes} 

\newcommand{\Bstates}{Q}
\newcommand{\Bstate}{q}
\newcommand{\Bstart}{\Bstate_0}
\newcommand{\Bnl}{\wnl} 
\newcommand{\Btf}{\wtf} 

Interpolation works as follows. We start with two EHMMs $\Q_\nsw$ and
$\Q_\swi$ on state space $\Bstates$. We interpret $\Q_\nsw(\rv
\Bstate_{i+1} | \rv \Bstate_i)$ as a specification of the state
evolution under normal ($\nsw$) behaviour, while we regard
$\Q_\swi(\rv \Bstate_{i+1} | \rv \Bstate_i)$ as a codification of what
happens upon a switch ($\swi$). The decision whether switches are
taken is left to a third EHMM $\C$, now on state space $\Cstates$,
with node labels $\Cnl$ in $\set{\nsw,\swi}$ that we use to select
which evolution is desired at each time step. The resulting
interpolation $\R = \Q_\nsw \otimes_\C \Q_\swi$ is displayed in
\figref{fig:BNinterpolation} as a Bayesian network on variables $\rv
\Cstate_i$, $\rv \sigma_i$, $\rv \Bstate_i$, $\rv \xi_i$ and $\rv x_i$
for each time $i=1,2,\ldots$ (Note that $\C$ differs from regular
EHMMs in that it does not have a data layer (the $\rv x_i$
variables), and the produced symbols are from $\{\nsw,\swi\}$ rather
than the fixed expert set $\Xi$.) We will now define the interpolation
distribution on these variables. Most conditional distributions are
copied from the input EHMMs: as before $\rv{\xi}_i =
\Bnl(\rv{\Bstate}_i)$, $\rv{\sigma}_i = \Cnl(\rv{\Cstate}_i)$, the
state evolution for the selection process $\R(\rv \Cstate_{i+1} | \rv
\Cstate_i)$ is copied from $\C$, and $\R(\rv x_i | \rv \xi_i)$ denotes
the prediction of expert $\rv \xi_i$ for the $i$th outcome.  As
indicated, the switch decisions $\rv \sigma_i = \Cnl(\rv{\Cstate}_i)$
are made between rounds, with $\rv \sigma_i = \swi$ indicating that a
switch occurs between time $i$ and $i+1$. In the interpolation the
probability of a state $\rv \Bstate_{i+1}$ now depends not only on the
previous state $\rv \Bstate_i$, but also on the switch decision $\rv
\sigma_i$, which determines which of the original two dynamics is
selected:
\[
\R(\rv \Bstate_{i+1} | \rv \Bstate_i, \rv \sigma_i)
~\df~
\Q_{\rv \sigma_i}(\rv \Bstate_{i+1} | \rv \Bstate_i)
.
\]
In addition, we have to define which of the two dynamics is used initially. We arbitrarily choose to define the model to start with the switching dynamics, and set
$\R(\rv \Bstate_1) \df \Q_\swi(\rv \Bstate_1)$.
We note that the interpolation $\R$ is again a prediction strategy of the form shown in \figref{fig:hmm}, now with joint state space $Q \times C$.

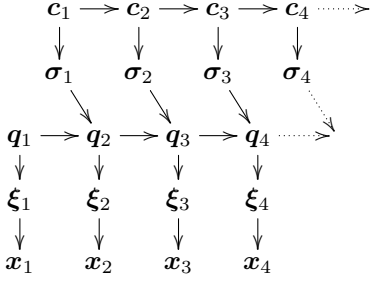
\begin{figure}
\caption{Bayesian network of interpolation $\R = \Q_\nsw \otimes_\C \Q_\swi$ \label{fig:BNinterpolation}}
\centering
$\xymatrix@!0@C=1.5em{
& {\rv \Cstate_1} \ar[rr] \ar[d] && {\rv \Cstate_2} \ar[rr]\ar[d] && {\rv \Cstate_3} \ar[rr]\ar[d] && {\rv \Cstate_4} \ar@{.>}[rr]\ar[d] &&
\\
& {\rv \sigma_1} \ar[dr] && {\rv \sigma_2} \ar[dr] && {\rv \sigma_3} \ar[dr] && {\rv \sigma_4} \ar@{.>}[dr] &&
\\
{\rv \Bstate_1} \ar[rr]\ar[d] && {\rv \Bstate_2} \ar[rr]\ar[d] && {\rv \Bstate_3} \ar[rr]\ar[d] && {\rv \Bstate_4} \ar@{.>}[rr]\ar[d] &&
\\
{\rv \xi_1} \ar[d]  && {\rv \xi_2} \ar[d]  && {\rv \xi_3} \ar[d]  && {\rv \xi_4} \ar[d]  &&
\\
{\rv x_1}  && {\rv x_2}  && {\rv x_3}  && {\rv x_4}  &&
}$
\end{figure}

Interpolation separates concerns; $\C$, on the highest level, determines when to switch. Below, $\Q_\nsw$ and $\Q_\swi$ determine the normal and switching behaviour. This separation is reflected in the following modular loss bound:

\begin{lemma}[Interpolation Decomposition]\label{lem:interpolation.loss.bound}
Fix interpolation $\R = \Q_\nsw \otimes_\C \Q_\swi$. For each sequence $\sigma^{t-1} \in \set{\nsw, \swi}^{t-1}$ of switch decisions (on the $\C$ level)
and each 
sequence  $\Bstate^t \in \Bstates^t$ of productive states (on the $\Q$ level)
\begin{align*}
\R(\Bstate^t)
&~\ge~
\C(\sigma^{t-1})~
\Q_\swi(\Bstate_1) \prod_{i=1}^{t-1} \Q_{\sigma_i}(\Bstate_{i+1} | \Bstate_i).
\end{align*}
\end{lemma}

\begin{IEEEproof}
By dropping terms from the marginal, we find
$\R(\Bstate^t) 
\ge
\R(\sigma^{t-1}) \R(\Bstate^t|\sigma^{t-1})$.
By the definition of interpolation, we have $\R(\sigma^{t-1}) = \C(\sigma^{t-1})$ and $\R(\Bstate^t|\sigma^{t-1}) = \Q_\swi(\Bstate_1) \prod_{i=1}^{t-1} \Q_{\sigma_i}(\Bstate_{i+1} | \Bstate_i)$.
\end{IEEEproof}

\begin{corollary}[Default Interpolation Regret]\label{cor:interbound}
  We apply this lemma to our $\Q$-level EHMMs of interest, $\Q_\nsw=\ehmm{b}[w]$ and $\Q_\swi=\ehmm{em}[w]$ where $w$ is uniform on $k$ experts. Fix $\xi^t$. Set $\sigma_i = \swi$
iff $\xi_{i+1} \neq \xi_i$, and let $m$ be the number of blocks in $\xi^t$, i.e.\ the number of $\swi$ in $\sigma^{t-1}$ plus one. Then for all data $x^t$
\[
\ln \frac{\ehmm{P}[\xi^t](x^t)}{\R(x^t)}
~\le~
- \ln \C(\sigma^{t-1}) + m \ln k.
\]
\end{corollary}
\begin{IEEEproof}Recall that in both $\Q$-level EHMMs there is, at each time, a one-one
  correspondence between productive states and experts, so we may just
  as well identify them. Then we have $\Q_\swi(\xi_1)=w(\xi_1)$, and
\[
\ehmm{b}[w](\xi_i = \xi_{i-1} | \xi_{i-1})= 1;\quad
\ehmm{em}[w](\xi_i | \xi_{i-1})= w(\xi_i).
\]
Now using $w(\xi) = 1/k$, \lemref{lem:interpolation.loss.bound} yields
$\R(\xi^t)\ge\C(\sigma^{t-1})~ k^{-m}$
and the result follows by \eqref{eq:bound1}.
\end{IEEEproof}

\begin{example}[Fixed Share Regret]
  We abbreviate $\ehmm{em}[\set{\nsw,\swi},\del{1-\alpha,\alpha}]$
  to $\ehmm{fs}[\alpha]$. We now redefine the Fixed Share model in
  terms of the interpolation
\[ \ehmm{fs}[w,\alpha] ~\df~
  \ehmm{b}[w] \otimes_{\ehmm{fs}[\alpha]}
  \ehmm{em}[w].
\]
This definition is equivalent to the one in~\figref{graph:mixedshare},
as the sets of infinite state sequences are in one-one correspondence
between the models.
We now reprove the Fixed Share regret bound
\eqref{eq:fixedshare.lossbound} by combining \corref{cor:interbound}
with the observation that
\[\begin{split}
- \ln
\ehmm{fs}[\alpha](\sigma^{t-1})&=-\ln\del[\big]{(1-\alpha)^{t-m}
  \alpha^{m-1}}\\
&=(t-1)\ent(\alpha^*,\alpha).
\end{split}\]\vskip-1.3\baselineskip \qedex
\end{example}\smallskip

In the following we often use this mechanism. We prove a loss bound
for the interpolator $\C$ on switch sequences, and transport it to the
data level using \corref{cor:interbound}, adding $m \ln k$. 

\subsubsection*{Detailed construction}
We now describe an explicit construction for the interpolation of two
EHMMs, and relate the size of the interpolation to the size of the
original models. This section is technical, and may be skipped on
first reading.

The construction takes three EHMMs, $\C$, $\Q_\nsw$ and $\Q_\swi$. The
\markdef{interpolator} $\C$ specifies \emph{when} switches occur,
$\Q_\nsw$ determines the \emph{normal} (\nsw) evolution and $\Q_\swi$
determines the evolution when a \emph{switch} (\swi) occurs.
In particular, we obtain a model that defines the same distribution as
$\ehmm{fs}[w,\alpha]$ by interpolation using $\C=\ehmm{em}[\set{\nsw,\swi},\del{1-\alpha,\alpha}]$, $\Q_\nsw=\ehmm{b}[w]$ and $\Q_\swi=\ehmm{em}[ w]$.

\begin{definition}[Interpolation] See~\figref{fig:interpolation} for an illustration.
Let $\C =
\tuple{
\Cstates, 
\Cstates^\prd,
\Cstart,
\Ctf,
\Cnl
}
$ be a EHMM on $\set{\nsw,\swi}$, and let $\Q_\nsw$ and $\Q_\swi$ be EHMMs on experts $\Xi$ sharing a common state set $\Bstates$ with identical start state $\Bstart$ and labelling $\Bnl$ (and thus identical productive states $\Bstates^\prd$), but with different transition functions $\Btf^\nsw$ and $\Btf^\swi$. 
For the construction below it will be convenient to prefix a new productive start state $\Cfake$ to $\C$ labelled $\Cnl(\Cfake) = \swi$,  with a single transition $\Ctf(\Cfake \to \Cstart) = 1$ to the original start state. This will ensure that we initially follow the $\Q_\swi$ dynamics.

We define $\Q_\nsw \otimes_\C \Q_\swi$, the \markdef{$\C$-interpolation of $\Q_\nsw$ and $\Q_\swi$}, by
\[
\Q_\nsw \otimes_\C \Q_\swi
~\df~
\tuple[\big]{
\Rstates, 
\Rstates^\prd,
\Rstart,
\Rtf,
\Rnl}
.\]
Each state of the interpolation is a pair of states, consisting of one
state from the common state set $\Bstates$, and one state from the
interpolator $\C$, at least one of them productive, plus a bit that indicates which state is to evolve next:
\[
\Rstates 
~\df~ 
\Bstates^\prd \times \Cstates \times \set{0}
\,\cup\,  
\Bstates \times \Cstates^\prd \times \set{1}.
\]
The productive states of the interpolation are the triples with two
productive states and zero bit
\[
\Rstates^\prd
~\df~
\Bstates^\prd \times \Cstates^\prd \times \set{0}
,
\]
and the interpolator has start state $\Rstart \df \tuple*{\Bstart, \Cfake, \indicator{\Bstart \not \in \Bstates^\prd}}$.
The transition function $\Rtf$ alternately forwards the components of the state, as indicated by the bit. First it evolves the first state component ($\Bstate$ in $\Bstates$) to the next productive state in $\Bstates$ using either $\Btf^\nsw$ or $\Btf^\swi$ as determined by the produced label $\Cnl(\Cstate) \in \set{\nsw, \swi}$. Then it forwards the second state component ($\Cstate$ in $\Cstates$) to the next productive state using $\Ctf$. That is
\[
\Rtf\del*{
\begin{aligned}
\tuple{\Bstate, \Cstate, 1} &\to \tuple{\Bstate', \Cstate, \indicator{\Bstate' \not\in \Bstates^\prd}}
\\
\tuple{\Bstate, \Cstate, 0} &\to \tuple{\Bstate, \Cstate', \indicator{\Cstate' \in \Cstates^\prd}}
\end{aligned}
} ~\df~
\del*{
\begin{aligned}
\Btf^{\Cnl(\Cstate)}(\Bstate &\to \Bstate')
\\
\Ctf(\Cstate &\to \Cstate')
\end{aligned}
}
\]
Finally, the node label is that of the first component
\[
\Rnl(\Bstate, \Cstate, 0)
~\df~
\Bnl\del{\Bstate}
.
\]
\figref{fig:interpolation.hmm} shows the state transition diagram of an interpolation, with the interpolator shown in \figref{fig:switc.dist.interpolator}. 
\begin{figure*}
\caption{Interpolation example: state transition diagram}\label{fig:interpolation}
\centering
\subfloat[EHMM $\C$ on $\set{\nsw,\swi}$ (normal/switch)\label{fig:switc.dist.interpolator}]{
$\xymatrix@!0@R=4em@C=2.3em{
    & \ex{\nsw} \ar[rr] & &\ex{\nsw} \ar[rr] & &\ex{\nsw} \ar@{.>}[r] & \\
&&&&&&\\
 \bn\ar[uur]\ar[r]\ar[ddr]   & \ex{\swi}  \ar[uurr] \ar[dr] & & \ex{\swi}  \ar[uurr] \ar[dr] && \ex{\swi} \ar@{.>}[dr] \ar@{.>}[ur] \\
 & & \si \ar[ur] \ar[dr] & & \si \ar[ur] \ar[dr] &&\\
    & \ex{\nsw} \ar[ur] & & \ex{\nsw} \ar[ur] && \ex{\nsw} \ar@{.>}[ur] &\\
}$
}
\qquad\qquad
\subfloat[{Interpolation $\Q_\nsw \otimes_\C \Q_\swi$ where $\Q_\nsw = \ehmm{b}[w]$ and $\Q_\swi = \ehmm{em}[w]$ on experts $\set{\expname{A}, \expname{B}, \expname{C}}$\label{fig:interpolation.hmm}}]{
$\xymatrix@!0@R=1.4em@C=2.7em{
&&    \si\ar[rr]       && \ex{a} \ar[rr] & & \si\ar[rr] && \ex{a} \ar[rr] & & \si\ar[r] & \si\\
&&   \si\ar[rr] && \ex{b} \ar[rr] & & \si \ar[rr] && \ex{b}  \ar[rr] & & \si\ar[r] & \si \\
&&    \si\ar[rr]      && \ex{c} \ar[rr] & & \si\ar[rr] && \ex{c} \ar[rr] & & \si\ar[r] & \si \\
\\
& \ex{a}\ar[uuuur]\ar[r]\ar@/^/[ddddddr] & \si\ar[dr] && \ex{a} \ar[uuuurr] \ar[dddr] & & \si\ar[dr] && \ex{a} \ar[uuuurr] \ar[dddr] & & \si\ar[dr] & \\
\bn\ar[ur]\ar[r]\ar[dr]& \ex{b}\ar[uuuur]\ar[r]\ar@/^/[ddddddr]  & \si\ar[r] &\si\ar[ur]\ar[r]\ar[dr] & \ex{b}\ar[uuuurr] \ar[dddr]  & & \si\ar[r] & \si\ar[ur]\ar[r]\ar[dr] & \ex{b} \ar[uuuurr] \ar[dddr]  & & \si \ar[r] & \si \\
&\ex{c}\ar[uuuur]\ar[r]\ar@/^/[ddddddr]&   \si\ar[ur] && \ex{c} \ar[uuuurr] \ar[dddr]  & & \si\ar[ur] && \ex{c} \ar[uuuurr] \ar[dddr]  & & \si\ar[ur] &  \\
& & & & &\si \ar[uuur] \ar[dddr] & & & &\si \ar[uuur] \ar[dddr]\\
& & & & &\si\ar[uuur] \ar[dddr] & & & &\si\ar[uuur] \ar[dddr]\\
& & & & &\si\ar[uuur] \ar[dddr] & & & &\si\ar[uuur] \ar[dddr]\\
&&    \si\ar[rr]     && \ex{a} \ar[uuur]  & &  \si\ar[rr] && \ex{a}  \ar[uuur]  & &  \si\ar[r] & \si\\
&    & \si\ar[rr] && \ex{b} \ar[uuur]  & & \si\ar[rr] && \ex{b}  \ar[uuur]  & & \si \ar[r] & \si\\
&&  \si\ar[rr]     && \ex{c} \ar[uuur]  & & \si\ar[rr] && \ex{c} \ar[uuur]  & & \si\ar[r] & \si \\
& \ar@{(-|}[l]+0^{\tuple{\cdot, \cdot, 1}}
& \ar@{(-|}[l]+0^{\tuple{\cdot, \cdot, 0}}
&&\ar@{(-|}[ll]+0^{\tuple{\cdot, \cdot, 1}}
&&\ar@{(-|}[ll]+0^{\tuple{\cdot, \cdot, 0}}
&&\ar@{(-|}[ll]+0^{\tuple{\cdot, \cdot, 1}}
&&\ar@{(-|}[ll]+0^{\tuple{\cdot, \cdot, 0}}
& \ar@{|-|}[l]+0^{\tuple{\cdot, \cdot, 1}}
}$}
\end{figure*}
\end{definition}

As mentioned before, the predictions of an EHMM can be computed with constant work per edge, where an edge is defined as a pair of states $\wstate, \wstate'$ with non-zero transition probability $\wtf(\wstate \to \wstate') > 0$.  We now bound the number of edges of an interpolation.

\begin{lemma}
Let $e^\C$, $e^\swi$, $e^\nsw$ and $e^{\R}$ be the numbers of edges in EHMMs $\C$, $\Q_\nsw$, $\Q_\swi$ and the interpolation $\R = \Q_\nsw \otimes_\C \Q_\swi$. Then
\[
e^\R
~\le~
\card{\Cstates^\prd} \max \set*{ e^\nsw, e^\swi}
+
\card{\Bstates^\prd} e^{\C}
\]
\end{lemma}
\begin{IEEEproof}
From the definition of $\Rtf$.
\end{IEEEproof}

This theorem provides an upper bound, for not all these edges may be reachable from the start state. A careful counting for Fixed Share yields that between the \emph{reachable} productive states at time $t$ and $t+1$ sit $c \card{\Xi}$ edges, for some constant $c$ independent of $t$. The running time of the interpolation version is hence no worse than that of the classical version.

This concludes the intermezzo. In the remainder of this section, we
discuss the benefits and costs of several choices for $\C$, both in
terms of loss bound and in terms of running time. We also briefly
discuss alternatives for $\Q_\nsw$ and $\Q_\swi$.

\subsection{Decreasing Switching Rate}\label{sec:dsr}

Fixed Share uses a fixed switching rate $\alpha$. However it is
possible to get good bounds without having to choose $\alpha$, by
letting the switching probability decrease as a function of time. This
trick was employed in the source coding setting in
\cite{ShamirMerhav1999}. 

Intuitively, if the number of blocks $m$ is small, then the term $(t-1)\ent(\alpha^*)$ in the Fixed Share bound \eqref{eq:fs.sandwich} with the optimal switching rate is not much smaller than $(m-1) \ln (t-1)$. To ensure at most an additive $\ln(t)$ penalty to the regret per switch, a switching rate of $\alpha_t = 1/t$ suffices. This inspires the models described in this section.

Whereas Fixed Share uses the elementwise mixture interpolator with
switching rate $\alpha$, we consider a new interpolator,
$\ehmm{dsr}[\alpha^\omega]$, which is similar to
$\ehmm{fs}[\alpha]$, except that the switching probability
$\alpha_t$ is no longer a parameter of the model, but a fixed
decreasing function of the time $t$. We still model switches as
independent, and as before, we define the full model as
\[\ehmm{dsr}[w,\alpha^\omega]~\df~\ehmm{b}[w]\otimes_{\ehmm{dsr}[\alpha^\omega]} \ehmm{em}[w].
\]

To obtain bounds, we use the following equality. Let $\sigma^{t-1}$ be
a sequence with $m-1$ occurrences of $\swi$ at positions $t_2, \ldots,
t_m$, and let $t_1=0$. Then
\begin{align}
&- \ln \ehmm{dsr}[\alpha^\omega](\sigma^{t-1}) 
=
- \ln \del[\Bigg]{\prod_{i=1}^{t-1} (1-\alpha_i) \prod_{j=2}^{m} \frac{\alpha_{t_j}}{1-\alpha_{t_j}}}\notag\\
&\quad=
- \sum_{i=1}^{t-1} \ln(1-\alpha_i) - \sum_{j=2}^{m}
\ln\frac{\alpha_{t_j}}{1-\alpha_{t_j}}\label{eq:bnd}.
\end{align}
The last expression can be read as follows: the first sum denotes the
cost of not switching during the first $t$ outcomes, and the second
sum denotes the correction for the switches that actually did occur.

Recall the two problems we identified above for Fixed Share: that the
switching rate $\alpha$ has to be tuned to obtain a good bound, and
that the regret keeps increasing even if, from some point on, no
switches occur anymore. Below we describe a first choice for
$\alpha_i$ that adequately solves the first of these two issues.  In
the next section we propose a different choice for $\alpha_i$ that
solves both problems simultaneously at the cost of a slight additional
overhead in the regret bound; a variant of this second model was shown
in~\cite{ErvenGrunwaldDeRooij2012} to yield a substantial improvement
of Bayes factors model selection; this will be discussed in more
detail below.
\subsubsection{Switching with Slowly Decreasing Probability}\label{sec:switch.slow}
%
\begin{theorem}
  Let $\alpha_i=1-e^{-c/i}$ for some $c>0$. Let $w$ be the uniform
  distribution on the set $\Xi$ of $k$ experts. For any data $x^t$ and 
  expert sequence $\xi^t$ with $m$ blocks
\begin{equation}\label{eq:switch.slow.bound}\begin{split}
&\ln\frac{\ehmm{P}[\xi^t](x^t)}{\ehmm{dsr}[w,\alpha^\omega](x^t)}\\
&\le m\ln k + c - (m-1) \ln c + (m-1 + c) \ln(t-1).\end{split}
\end{equation}
\end{theorem}
\begin{IEEEproof}
By
  \eqref{eq:bnd}, using $\sum_{i=1}^{t} \frac{1}{i} < \ln t + 1$ and
  $e^x \ge x+1$,
\begin{align}
\notag
&- \ln \ehmm{dsr}[\alpha^\omega](\sigma^{t-1}) 
=
c \sum_{i=1}^{t-1} \frac{1}{i} - \sum_{j=2}^{m} \ln \del[\big]{e^{c/t_j}-1}\\
\label{eq:logtvisible}
&\le
c \ln (t-1) + c - (m-1) \ln c + \sum_{j=2}^{m} \ln t_j.
\end{align}
The sum is bounded by substituting each $t_j$ by
$t-1$, and the result follows by \corref{cor:interbound}.
\end{IEEEproof}
Note that while we succeeded in eliminating the parameter $\alpha$, we
have in fact introduced a new parameter $c$, so it would appear that
matters have not improved much. But unlike $\alpha$, a suboptimal
value for $c$ only yields a regret penalty of order $\ln t$, so it
may safely be set to some convenient constant like $c=1$. The
optimising value is $c^* = (m-1)/(1 + \ln(t-1))$, which yields
slightly better asymptotics, but this defeats the purpose as it would
require a priori knowledge of $m$ and $t$ again.

We now compare the regret bound~\eqref{eq:switch.slow.bound} to the
bound \eqref{eq:fixedshare.lossbound} for Fixed Share. To maximise
the difference, we use the optimising parameter $\alpha^*$ for Fixed
Share, and we lower bound the entropy using~\eqref{eq:fs.sandwich}.
The difference is
\[
c-(m-1)\ln c+c\ln(t-1)+(m-1)\ln(m-1),
\]
where the last two terms dictate asymptotic behaviour. Which of these
terms is dominant depends on how quickly $m$ grows as a function of
$t$. If there are relatively few switches, $m\ln m=o(\ln t)$, then the
$c\ln(t-1)$ term dominates, so it pays to use a small value for $c$ to
get good asymptotics in that case. If, on the other hand, the number
of switches is large, then the last term is larger, and it may be
substantial; careful judgement is then required to decide whether or
not this is an acceptable price to pay or that a more sophisticated
method for \emph{learning} the switching rate (\secref{sec:learnrate}) is
preferable.

\subsubsection{Switching with More Quickly Decreasing
  Probability}\label{sec:switch.fast}
In some settings the optimal number of switches between experts may
remain bounded. A natural example is Bayes factors model selection,
where the considered experts are Bayesian prediction strategies
associated with model classes of varying complexity; at small sample
sizes, simple model classes typically yield the best predictions (as
their parameter estimates are quicker to converge to their optimal
values), but if one of the more complex model classes contains the
data generating distribution, then that model class eventually
produces the best predictions. From that point in time onwards, no
more switches away from that model class are required.

In such a scenario, a simple Bayesian mixture of the experts with
uniform prior yields a regret bound of $\ln k$ w.r.t.\ the ultimately
best expert (see \eqref{eq:drop_terms}), which depends on the number
of experts but \emph{not on the sample size}. Asymptotically, this is
therefore a better solution than the one presented in the previous
section, where even if there are no switches at all ($m=1$), the
incurred regret bound of $\ln k+c+c\ln(t-1)$ grows without bound due
to the first term of \eqref{eq:logtvisible}. This happens because the
ES-prior $\ehmm{dsr}[\alpha^\omega]$ assigns zero probability to the
event that no more switches occur from some time $t$ onwards. The
problem with the Bayesian mixture, as apparent in \eqref{eq:drop_terms}, is that it cannot take advantage of
the superior performance of the simpler models at small sample
sizes. As shown in~\cite{ErvenGrunwaldDeRooij2012}, this results in a
suboptimal rate of convergence in the nonparametric case for the
Bayesian mixture: its overhead compared to a switching model (such as
the one from the previous section) can be arbitrarily large!

To achieve the best of both worlds, we must tweak the model from the
previous section somewhat: while we still assign positive prior
probability to the occurrence of switches, we also ensure that the
probability that no more switches occur from any given time onwards is
strictly positive. This section describes a simplification of the
\emph{Switch Distribution}\footnote{Incidentally, the switch
  distribution is a tracking model whose interpolator has the
  structure depicted in \figref{fig:switc.dist.interpolator}. The idea
  is that with every switch, there is a certain fixed probability of
  ``stabilisation'', meaning that the interpolator enters a special
  ``band'' of states where further switching is impossible.} proposed
in~\cite{ErvenGrunwaldDeRooij2012} for which the results of that paper
still hold. In brief, this model achieves the optimal rate of risk
convergence when used for sequential prediction, but at the same time,
it defines a \emph{consistent} model selection criterion (it selects the
model containing the true distribution with probability $1$ as
sufficient data become available).

\begin{theorem} 
  Let $\alpha_i=1-e^{-c\pit(i)}$ for some $c>0$ and a decreasing probability mass
  function $\pit$ on the positive integers.  Let $w$ be the uniform
  distribution on the set $\Xi$ of $k$ experts.  For any data $x^t$ and
  expert sequence $\xi^t$ with $m$ blocks
\[
\ln\frac{\ehmm{P}[\xi^t](x^t)}{\ehmm{dsr}[w,\alpha^\omega](x^t)}
\le m\ln k+
c - (m-1) \ln c - (m-1)\ln\pit(t_m).
\]
\end{theorem}
\begin{IEEEproof}
Using \eqref{eq:bnd},
  $\sum_i\pit(i) = 1$ and $e^x \ge x+1$,
\[\begin{split}
&- \ln \ehmm{dsr}[\alpha^\omega](\sigma^{t-1}) 
=
c \sum_{i=1}^{t-1}\pit(i) - \sum_{j=2}^{m} \ln \del[\big]{e^{c\cdot\pit(t_j)} -1}\\
&\le
c - (m-1) \ln c - \sum_{j=2}^{m} \ln\pit(t_j).
\end{split}\]
For decreasing $\pit$, we obtain an upper bound by substituting
$t_i=t_m$ for $1\le i<t_m$, and the theorem follows 
from \corref{cor:interbound}. 
\end{IEEEproof}
A desirable feature of this bound is that it is expressed in terms of
the index $t_m$ of the last switch rather than in terms of the time
$t$, as we obtained in \secref{sec:switch.slow}. The role of $c$ is
even weaker than before, as there is no $c\ln t$ penalty term; its
optimal value is now $c^*=m-1$, but in practice $c=1$ would be a
sensible value. On the flip side, a $\ln t$ cost per switch as in \eqref{eq:switch.slow.bound}
can no longer be guaranteed. A convenient fat-tailed prior that comes very close is
\begin{equation}\label{eq:cool.prior}
  \pit(t)=\frac{1}{\ln(t+e-1)}-\frac{1}{\ln(t+e)},
\end{equation}
which satisfies $-\ln\pit(t)\le\ln(t)+2\ln\ln(t+e)+e/t$.

\subsection{Learning the Switching Rate}\label{sec:learnrate}
\subsubsection{The Switching Method}


In a very early publication \cite{volfwillems1998}, Volf and Willems
describe an algorithm called \emph{the switching method}, which is
very similar to Herbster and Warmuth's Fixed Share, except that it is
able to learn the optimal switching rate $\alpha$ on-line. A similar
method was developed independently in \cite{Vovk1999}. Here we
describe the Switching Method as an interpolation and bound its regret. Whereas Fixed
Share interpolates using a fixed $\text{Bernoulli}[\alpha]$
distribution, the switching method ``integrates out'' the parameter
using Jeffreys' prior (which is $\text{Beta}[\half,\half]$).

The switching method EHMM is defined as the interpolation
\[
\ehmm{sm}[w] 
~\df~
\ehmm{b}[w] 
\otimes_{\ehmm{sm}[]} 
\ehmm{em}[w],
\]
with the interpolator $\ehmm{sm}[]$ defined in
\figref{graph:universal.share}. Each productive state $\tuple{n_\nsw,
  n_\swi,\sigma}$ represents the fact that after observation $n_\nsw
+ n_\swi +1$ a switch occurs $(\sigma=\swi)$ or not $(\sigma=\nsw)$,
while there have been $n_\swi$ switches in the past.

We now bound the regret of the switching method with respect to Fixed
Share with any switching rate $\alpha$ (in particular the maximum
likelihood rate $\hat \alpha$), and thereby show that it is universal
for the Fixed Share model class $\setc[\big]{\ehmm{fs}[w,\alpha]}{
  \alpha \in \intcc{0,1}}$.

\begin{theorem}[The Switching Method Regret]\label{thm:sm.loss.bound}For any switching rate $\alpha$ and data $x^t$
\[
\ln \frac{\ehmm{fs}[w, \alpha](x^t)}{\ehmm{sm}[w](x^t)}
~\le~
\ln 2 + \tfrac{1}{2} \ln t.
\]
\end{theorem}

\begin{IEEEproof}
Fixed Share and the switching method interpolate the same EHMMs, so we
have the following information processing inequalities (c.f.\
\lemref{lem:toexpectation})
\[\begin{split}
&\max_{x^t} 
\frac{\ehmm{fs}[w, \alpha](x^t)}{\ehmm{sm}[w](x^t)}
\le
\max_{\xi^t} 
\frac{\ehmm{fs}[w, \alpha](\xi^t)}{\ehmm{sm}[w](\xi^t)}\\
&\le
\max_{\sigma^{t-1}} 
\frac{\ehmm{fs}[w, \alpha](\sigma^{t-1})}{\ehmm{sm}[w](\sigma^{t-1})}
=
\max_{\sigma^{t-1}} 
\frac{\ehmm{fs}[\alpha](\sigma^{t-1})}{\ehmm{sm}[w](\sigma^{t-1})}.
\end{split}\]
Thus we may transfer regret bounds from the interpolator level via the
expert-sequence level to the data level. The rightmost term is the
worst-case regret for the Bernoulli model with Jeffreys prior, which
can be bounded (see e.g.~\cite{willems96:coding_piecew}) by $\ln 2 +
\tfrac{1}{2} \ln t$ for all $\alpha$.
\end{IEEEproof}

\noindent
By the previous theorem and the Fixed Share regret bound
\thmref{thm:fixed.share.loss.bound}, we obtain for all $\xi^t$ with
switching frequency $\alpha^*$
\[
\ln \frac{\ehmm{P}[\xi^t](x^t)}{\ehmm{sm}[w](x^t)}
~\le~
m \ln k + 
(t-1) \ent(\alpha^*) + \ln 2 + \tfrac{1}{2} \ln t.
\]
The switching method was independently derived in~\cite{bousquet2003},
where this last inequality is also proved. Our theorem is slightly
sharper, as it bounds the regret w.r.t.\ the actual maximum-likelihood
Fixed Share performance instead of its regret bound.

\begin{figure}
\caption{The switching method interpolator $\ehmm{sm}[]$}\label{graph:universal.share}
\centering
\subfloat{%
\newcommand{\exA}{1}%
\newcommand{\exB}{0}%
$\xymatrix@R=0.5em@C=1.4em{
&
&&
&&
&&&
\\
&   & 
&   & 
&   & 
& \ex{\swi} \ar@{.>}[ur]  & 
\\
&
&&
&&
&\si \name{3,0} \ar[r] \ar[ur] & \ex{\nsw} \ar@{.>}[r] &
\\
& & 
& & 
& \ex{\swi} \ar[ur] & 
& \ex{\swi} \ar[ur] & 
\\
& 
&&
&\si \name{2,0} \ar[r] \ar[ur] & \ex{\nsw} \ar[r]
&\si \name{2,1} \ar[r] \ar[ur] & \ex{\nsw} \ar@{.>}[r] &
\\
& & 
& \ex{\swi} \ar[ur] & 
& \ex{\swi} \ar[ur] & 
& \ex{\swi} \ar[ur] & 
\\
& 
&\si \name{1,0} \ar[r] \ar[ur] & \ex{\nsw} \ar[r]
&\si \name{1,1} \ar[r] \ar[ur] & \ex{\nsw} \ar[r]
&\si \name{1,2} \ar[r] \ar[ur] & \ex{\nsw} \ar@{.>}[r] &
\\
& \ex{\swi} \ar[ur] & 
& \ex{\swi} \ar[ur] & 
& \ex{\swi} \ar[ur] & 
& \ex{\swi} \ar[ur] & 
\\
 \bn \name{0,0} \ar[r] \ar[ur] & \ex{\nsw} \ar[r]
&\si \name{0,1} \ar[r] \ar[ur] & \ex{\nsw} \ar[r]
&\si \name{0,2} \ar[r] \ar[ur] & \ex{\nsw} \ar[r]
&\si \name{0,3} \ar[r] \ar[ur] & \ex{\nsw} \ar@{.>}[r] &
}$%
}
\quad
\subfloat{%
\setlength{\abovedisplayskip}{0cm}
\setlength{\belowdisplayskip}{0cm}
\small%
\begin{minipage}[t]{.47\textwidth}
\begin{gather*}
\begin{aligned}
\wstates &= \wstates^\sil \cup \wstates^\prd
&
\wstates^\sil &= \nats^2
&
\wstates^\prd &= \nats^2 \times \set{\nsw,\swi}
\end{aligned}
\\
\begin{aligned}
\wstart &= \tuple{0,0}
&
\wnl(n_\nsw, n_\swi, \sigma) &= \sigma
\end{aligned}
\\
\wtf \del*{
\begin{aligned}
\tuple{n_\nsw, n_\swi,\nsw} &\to \tuple{n_\nsw+1, n_\swi}
\\
\tuple{n_\nsw, n_\swi,\swi} &\to \tuple{n_\nsw, n_\swi+1}
\\
\tuple{n_\nsw, n_\swi} &\to \tuple{n_\nsw, n_\swi, \nsw}
\\
\tuple{n_\nsw, n_\swi} &\to \tuple{n_\nsw, n_\swi, \swi}
\end{aligned}} =
\del*{
\begin{gathered}
1
\\
1
\\
\tfrac{(n_\nsw +\frac{1}{2})}{(n_\nsw + n_\swi+1)}
\\
\tfrac{(n_\swi +\frac{1}{2})}{(n_\nsw + n_\swi+1)}
\end{gathered}}
\end{gather*}
\end{minipage}
}
\end{figure}
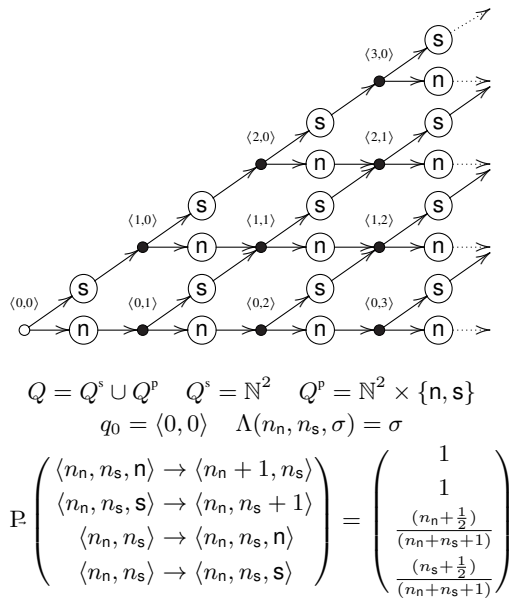%

\subsubsection{Improving Time Efficiency for Learning the Switching Rate}\label{sec:montjaak}
The new ingredient of the switching method compared to Fixed Share is
that the EHMM includes a switch count in each state. This allows us to
adapt the switching probability to the data, but it also renders the
number of states quadratic. The quadratic running time $O(k\,t^2)$
restricts its use to moderately sized data sets. The approach taken by Monteleoni and Jaakkola~\cite{Jaakkola2003} is to place a
\emph{discrete} prior on the switching rate $\alpha$: the prior mass is distributed over $\sqrt t$
well-chosen points, where the ultimate sample size $t$ is assumed
known. This way they still achieve the bound of
\thmref{thm:sm.loss.bound} up to a constant, while reducing the
running time to $O(k\,t\sqrt t)$.

This approach has two disadvantages of
its own: first, the ultimate sample size $t$ has to be known in
advance, which means that the presented algorithm is only
quasi-online. Second, the discretisation of the prior is obtained by a numeric optimisation procedure, which means that both the number and the locations of
the discretisation points are not known in closed form. As a
consequence, the resulting regret bound can only be determined up to
$O(1)$. In~\cite{DeRooijVanErven09} a simple \emph{explicit}
discretisation scheme is presented which allows the regret bound to be
calculated exactly. Furthermore, it is shown how, at the cost of a
somewhat worse regret bound, this discretisation scheme can be
\emph{refined} online such that $t$ no longer has to be known in
advance.

\subsection{The Run-length Model for Clustered Switching}\label{sec:rl}
%

Run-length codes have been used extensively in the context of data
compression, see e.g.~\cite{Moffat2002}. The corresponding probability distributions are known in the statistical literature as renewal processes, see \cite{renewal}.

Rather than applying run
length codes directly to the observations, we use them as interpolators, as
they constitute good models for the distances between consecutive
switches.

The run-length model is especially useful if the switches are
clustered, in the sense that some parts of the expert sequence contain
relatively few switches, while other parts contain many. The Fixed
Share algorithm remains oblivious to such properties, as its
interpolator is a Bernoulli model: the probability of switching
remains the same, regardless of the index of the previous
switch. Essentially the same limitation also applies to the switching
method, whose switching probability normally converges as the sample
size increases. The decreasing $\alpha$ models perform well when the
switches are clustered toward the beginning of the sample, but
depending on the application this may be unrealistic and may introduce
a new unnecessary loss overhead.

The run-length model, which is related to Willems' ``linear complexity
coding method'' from~\cite{willems96:coding_piecew} and its subsequent refinement in~\cite{ShamirMerhav1999}, models
the \emph{intervals} between successive switches as
independently distributed according to some distribution $\pit$. After
the switching method and decreasing $\alpha$ models, this is a third
generalisation of the Fixed Share algorithm, which is recovered by
taking a geometric distribution for $\pit$: the interpolation then becomes memoryless and reduces to the interpolator of Fixed Share.

Let $\pit$ be a distribution on $\posints \cup \set{\infty}$, which is
used to model the lengths of the blocks. We assume $\pit(\infty)>0$;
this keeps our regret constant when the reference number of switches
is bounded while the number of samples goes to infinity. The
run-length interpolator $\ehmm{rl}[\pit]$ is defined in
\figref{fig:dfe.automaton}. Intuitively, the state $\tuple{t,\delta}$
means that we are at time $t$, and that sample $t+1$ will be the
$\delta$th sample since the last switch.  
The EHMM for the run-length model is given by the interpolation
\[
\ehmm{rl}[w, \pit] 
~\df~
\ehmm{b}[w] 
\otimes_{\ehmm{rl}[\pit]} 
\ehmm{em}[w].
\]
As may be read from the diagram of the interpolator, we require
quadratic running time $O(k\,t^2)$ to evaluate the run-length model in
general.

\begin{figure}
\caption{The run-length model interpolator $\ehmm{rl}[\pit,c]$}\label{fig:dfe.automaton}
\centering
\subfloat{%
$\xymatrix@R=0.5em@C=1.4em{
& &   &  &  &  & &&
\\
 &  &  &  &  & &&\ex{\swi} \ar@{.>}[ur] &
\\
& &   &  &  &  & \si \bname{3,1} \ar[r]\ar[ur] & \ex{\nsw} \ar@{.>}[r] &
\\
 &  &  &  &  & \ex{\swi} \ar[ur] &
\\
& &   &  & \si \bname{2,1} \ar[r]\ar[ur]  & \ex{\nsw} \ar[r] & \si \bname{3,2} \ar[r] \ar[uuur] & \ex{\nsw} \ar@{.>}[r] &
\\
 &  &  & \ex{\swi} \ar[ur] &  & &
\\
& & \si \bname{1,1} \ar[r]\ar[ur]  & \ex{\nsw} \ar[r] & \si \bname{2,2} \ar[r]\ar[uuur]  & \ex{\nsw} \ar[r] & \si \bname{3,3} \ar@/^.5ex/[uuuuur] \ar[r] & \ex{\nsw} \ar@{.>}[r] &
\\
 & \ex{\swi} \ar[ur] &  &  &  &  &
\\
\bn \bname{0,1} \ar[r] \ar[ur] & \ex{\nsw} \ar[r]  & \si \bname{1,2} \ar[r] \ar[uuur] & \ex{\nsw} \ar[r] & \si \bname{2,3} \ar[r] \ar@/^.5ex/[uuuuur] & \ex{\nsw} \ar[r] & \si \bname{3,4} \ar[r] \ar@/^1ex/[uuuuuuur] & \ex{\nsw} \ar@{.>}[r] &
}$%
}
\quad
\subfloat{%
\setlength{\abovedisplayskip}{0cm}
\setlength{\belowdisplayskip}{0cm}
\small%
\begin{minipage}[t]{.47\textwidth}
\begin{gather*}
\begin{aligned}
\wstates &= \wstates^\sil \cup \wstates^\prd
&
\wstates^\sil &= \mathbb S
&
\wstates^\prd &= \set*{\nsw} \times \mathbb S ~\cup\, \set*{\swi} \times \nats
\end{aligned}
\\
\begin{aligned}
\wstart &= \tuple{0,1}
&
\wnl(\nsw,t,\delta) &= \nsw
&
\wnl(\swi,t) &= \swi
\end{aligned}
\\
\wtf \del*{\!\!
\begin{aligned}
\tuple{\swi,t} & \to \tuple{t, 1}
\\
\tuple{\nsw, t, \delta} & \to \tuple{t,\delta}
\\
\tuple{t, \delta} & \to \tuple{\nsw, t\!+\!1, \delta\!+\!1}
\\
\tuple{t, \delta} & \to \tuple{\swi, t\!+\!1}
\end{aligned}\!} =
\del*{\!
\begin{gathered}
1
\\
1
\\
\pit({\scriptsize \rv z > \delta | \rv z \ge \delta})
\\
\pit({\scriptsize \rv z = \delta | \rv z \ge \delta})
\end{gathered}\!}
\end{gather*}
where
\[
\mathbb S \df \set*{\tuple{t,\delta} \in \nats^2 \mid \delta \le t+1}. 
\]
\end{minipage}
}
\end{figure}

\begin{theorem}[Run-length Model Regret]\label{thm:rlebound}
  Let $w$ be the uniform distribution on $k$ experts.  Assume there is
  a log-convex function $\vartheta$ on $[1,\infty)$ that agrees with
  $\tau$ on $\posints$. With abuse of notation, we identify $\tau$
  with $\vartheta$. Then, for all data $x^t$ and expert sequences
  $\xi^t$ with $m$ blocks, we have
\begin{equation}\label{eq:rlebound}
\ln \frac{\ehmm{P}[\xi^t](x^t)}{\ehmm{rl}[w, \pit](x^t)}
~\le~
m\ln k -\ln\pit(\infty)- (m-1)\ln\pit\del*{\frac{t_m}{m-1}}.
\end{equation}
\end{theorem}
\begin{IEEEproof}
  Fix a switch sequence $\sigma^{t-1}$ with $m-1$ occurrences of
  $\swi$ at positions $t_2,\ldots,t_m$, and let $t_1=0$. For $j=1,
  \ldots, m-1$, let $\delta_j=t_{j+1}-t_j$ denote the length of block
  $j$. From the definition of the interpolator above, we obtain
\[\begin{split}
- \ln \ehmm{rl}[\pit](\sigma^{t-1})
&=
-\ln\pit(\rv z \ge t-t_m)-\sum_{j=1}^{m-1}\ln \pit(\delta_j)\\
&\le
-\ln \pit(\infty)-\sum_{j=1}^{m-1}\ln \pit(\delta_j)
.
\end{split}\]
Since $-\ln \pit$ is concave, by Jensen's inequality we have
\[
\sum_{j=1}^{m-1}\frac{-\ln\pit(\delta_j)}{m-1}
\le -\ln\pit\del[\Bigg]{\sum_{j=1}^{m-1} \frac{\delta_j}{m-1}}
=-\ln\pit\del*{\frac{t_m}{m-1}}.
\]
In other words, the block lengths $\delta_i$ are all equal in the
worst case. Combining this with \corref{cor:interbound} we obtain the
result.
\end{IEEEproof}

We have seen that the run-length model reduces to Fixed Share if the
prior on switch distances $\pit$ is geometric, so it can be
evaluated in linear time in that case. The geometric prior is not the
only one for which the complexity can be reduced; for example, the
negative binomial distribution can also be implemented efficiently, as
well as any $\tau$ with finite support. However, such priors must have
exponentially small tails; for priors $\tau$ with thick tails, which
are desirable in our worst-case analysis, one may use the fully general
EHMM as depicted in Figure~\ref{fig:dfe.automaton}.

To compare the performance of the run-length model to the bound
\eqref{eq:fixedshare.lossbound} for Fixed Share, assume $t_m=t-1$ and
define $\pit$ as in~\eqref{eq:cool.prior}. The
bound~\eqref{eq:rlebound} becomes
\begin{multline*}
  m\ln k-\ln\pit(\infty)+(m-1)\ln\frac{t-1}{m-1}\\
  +2(m-1)\ln\ln\del*{\frac{t-1}{m-1}+e}+(m-1)e.
\end{multline*}
To maximise the difference, we use the optimising parameter $\alpha^*$
for Fixed Share, and we bound the entropy from below
using~\eqref{eq:fs.sandwich}.  The gap between the bounds is then
given by
\[-\ln\pit(\infty)+2(m-1)\ln\ln\del*{\frac{t-1}{m-1}+e}+(m-1)e.\]
At this modest price, the run-length model does not require tuning any
parameters, its regret depends on $t_m$ instead of $t$, and it may
take advantage of clustered switches, although this is not expressed
by \thmref{thm:rlebound}.

Willems and Krom showed in \cite{WillemsKrom1997} how the time complexity of the algorithm can be improved at a cost of a slightly worse regret bound. They force a switch at carefully chosen states in the run-length EHMM (\figref{fig:dfe.automaton});
all subsequent states in the same row become unreachable and hence can be pruned from the model. Using this scheme the number of states reachable in the model at time $t$ (and therefore the complexity per time step) can be reduced to $\log_2 t$, while the regret bound deteriorates, also by a logarithmic factor. This approach was later refined to obtain more general time complexity / regret tradeoffs in~\cite{ShamirMerhav1999,
  ShamirCostello2000, hazan2009efficient, GyorgyLinderLugosi12012}.

\subsection{Ordered Experts}\label{sec:oe}
%
In the models discussed so far, once a switch occurs, it is equally
easy to switch to any of the available experts, as $\Q_\swi$
prescribes uniform redistribution of the probability mass. This
approach is reasonable if we do not know anything about the
relationship between the experts; furthermore it has the advantage
that percolating probabilities through $\Q_\swi$ requires only $O(k)$
operations, while we would need $O(k^2)$ operations to support
arbitrary transition probabilities between the experts. In this
section we consider an interesting alternative that both makes
intuitive sense and allows for efficient computation.

Assume that the experts can be sensibly organised using a line or ring
topology, with the interpretation that switches between two experts
are more likely if they are close together on this structure than if
they are far apart.  An example is given by Vovk \cite{Vovk1999} who considers
a polynomial regression problem with one expert to represent the polynomials of each degree. In this
case it is clear that, typically, the optimal degree
increases gradually as more observations are gathered, so switching
from degree 10 to degree 11 is more likely than,
say, switching instead to degree 1,000.

We will simplify matters further by postulating that the probability
of a switch between any pair of experts who are $\delta$ apart is the
same. Furthermore, for simplicity of exposition we identify the
experts with the integers, $\Xi=\ints$. (In practice it is of course
not possible to work with an infinite set of experts, but this can be
resolved by simply changing the forward algorithm to drop all
probability mass that at any time becomes propagated to an expert
outside of the considered range.)

Now the notion of similarity between experts may be expressed by a
kernel $\kappa$, i.e.\ a probability distribution on distances. This
allows us to specify an expert sequence prior of the form
\[
\pi(\xi_{t+1}\mid \xi^t)=\kappa(\xi_{t+1}-\xi_t).
\]
Note that in principle a different kernel can be used every round, but
it suffices for our purpose to consider a fixed kernel. If an EHMM
$\Q$ has a marginal distribution on expert sequences of this form, it
is relatively easy to maintain the weights on the experts. A property of Hidden Markov Models is that
the next state $\rv \xi_{t+1}$ is independent of the observed data $\rv x^t$ given the current state $\rv \xi_t$, so that
$\ehmm{Q}[\pi](\xi_{t+1} \mid \xi_t, x^t) = \pi(\xi_{t+1} \mid \xi_t) = \kappa(\xi_{t+1}-\xi_t)$. Therefore we can compute the marginal distribution on the expert at time $t+1$ given the previously
observed data as
\[
\ehmm{Q}[\pi](\xi_{t+1} \mid x^t) 
=
\sum_{\xi_t} \kappa(\xi_{t+1} - \xi_t) \ehmm{Q}[\pi](\xi_t \mid x^t)
\fd
\kappa * \ehmm{Q}[\pi](\xi_t \mid x^t),
\]
where the asterisk denotes the \emph{convolution} operation.
This approach can be lifted to the level of states: sometimes it may
be sensible to order all states involved in an expert HMM. However,
for simplicity we will consider the interpolating model of
\secref{sec:interpolation}, where the transitions of $\Q_\swi$
are replaced by a convolution $\kappa$ on the experts. An EHMM
implementing such convolutions is
$\ehmm{kernel}[\kappa]\df\tuple{ \wstates, \wstates^\prd, \wstart,
  \wtf, \wnl }$, defined as follows
\begin{subequations}
\begin{gather*}
\begin{aligned}
\wstates &=\wstates^\prd \cup \wstates^\sil 
&
\wstates^\prd &=\ints\times\posints
&
\wstates^\sil &= \set{\tuple{0,0}}
&
\wstart &=\tuple{0,0}
\end{aligned}
\\
\begin{aligned}
  \wnl(\xi,t)&=\xi
&
\wtf(\tuple{\xi,t}\to\tuple{\xi',t+1})&=\kappa(\xi'-\xi)
\end{aligned}
.
\end{gather*}
\end{subequations}
For this scenario, we derive the following analogue of
\corref{cor:interbound}:
\begin{corollary}[Kernel Interpolation Regret]\label{cor:kernel.bound}
  Let $\Q_\nsw=\ehmm{b}[\ints, \kappa]$ and $\Q_\swi=\ehmm{kernel}[\kappa]$. Fix $\xi^t$. Set $\sigma_i=\swi$ iff
  $\xi_{i+1}\ne\xi_i$, and for $1\le j\le m$ let $k_j$ denote the
  expert used in the $j$th block. Further let $k_0=0$. Then for all $x^t$:
\[
\ln \frac{\ehmm{P}[\xi^t](x^t)}{\del*{\Q_\nsw \otimes_\C \Q_\swi}(x^t)}
~\le~
- \ln \C(\sigma^{t-1}) - \sum_{j=1}^m \ln\kappa(k_j-k_{j-1}).
\]
\end{corollary}
\begin{IEEEproof}As before, we identify productive states and experts to
  get $\ehmm{b}[\ints, \kappa](\xi_1) =
  \ehmm{kernel}[\kappa](\xi_1) = \kappa(\xi_1)$,
  $\ehmm{b}[\ints, \kappa](\xi_i = \xi_{i-1} | \xi_{i-1}) =
  1$, and $\ehmm{kernel}[\kappa](\xi_i | \xi_{i-1})=
  \kappa(\xi_i-\xi_{i-1})$.  Now
  \lemref{lem:interpolation.loss.bound} yields $\del*{\Q_\nsw
    \otimes_\C \Q_\swi}(\xi^t)\ge \C(\sigma^{t-1})~\prod_{j=1}^m
  \kappa(k_j-k_{j-1})$, and the result follows by \eqref{eq:bound1}.
\end{IEEEproof}
From the Convolution Theorem, we know that any convolution $\kappa *
\lambda$ on $k$ experts can be carried out in $O(k \ln k)$ time using
the Fast Fourier Transform algorithm, see e.g.
\cite{cootuk1965,bracewell65}. Thus, the ordered expert approach,
which can be combined with any of the interpolating models described
in previous sections, seems to provide a very attractive tradeoff
between time complexity and expressive power.

In the following we consider a particular kernel for which the
convolution can be performed in $O(k)$ time using a much simpler
algorithm. It also has an interesting interpretation as a nice model
for ``parameter drift''.

\subsection{Parameter Drift}\label{sec:drift}

So far, we have discussed strategies where we follow a Bayesian
prediction strategy which is interrupted every now and then by
switching events. This is reflected by the regret bound
\corref{cor:kernel.bound}, which consists of a term for the
cost of specifying the indices of the switches, and a second term for
the cost of specifying which experts are involved in the switches.

In this section we take a radically different approach. Rather than
thinking of sporadic abrupt changes in the relative predictive
performance of the experts, we now imagine that their performance
changes gradually over time. Sticking to the ordered experts approach,
as before we identify the set of experts with the integers,
$\Xi=\ints$ (indicating for example the number of bins in a regular histogram, or arising from discretisation of a continuous parameter space).  However, in this section we will bound the regret in
terms of the total amount of \emph{drift} in $\xi^t$:
\[
d=\sum_{i=1}^t \abs{\delta_i}\text{, where $\delta_1 = \xi_1$
  and $\delta_i=\xi_i-\xi_{i-1}$ for $1<i\le t$,}
\]
which can be viewed as the length of the path described by $\xi^t$.

As an example, one may consider the switching model proposed by
Monteleoni and Jaakkola (see \secref{sec:montjaak}). They
essentially instantiate a number of Fixed Share models, for various
values of the switching rate $\alpha$. These Fixed Share instances are
prediction strategies, and can therefore be interpreted as experts
themselves. However, it seems reasonable to assume that in many cases
the optimal switching rate $\alpha$ might be subject to drift: it
might vary somewhat as time progresses. Therefore it may be beneficial
to combine these ``Fixed Share experts'' using a model that can
represent parameter drift. The resulting loss can be bounded in terms
of the amount of drift that occurs in the reference sequence of
switching parameters.
For parameter drift we no longer use an interpolation, as in previous
sections, because switches no longer have special status.  Instead,
shifts between experts are possible at each time step, through
convolution with the following kernel, parameterised by $0<\alpha<1$:
\begin{equation}\label{eq:kappa.alpha}
\kappa_\alpha(\delta)~\df~\alpha^{|\delta|}\frac{1-\alpha}{1+\alpha}.
\end{equation}
This kernel can be implemented with the EHMM
$\ehmm{kernel}[\kappa_\alpha]$ from the previous section, but
as it turns out it is possible to represent the same kernel using a
different EHMM $\ehmm{pd}[\alpha]$, defined in
\figref{fig:geometric.kernel.update}, that uses silent states to
reduce the number of edges, allowing the convolution to be carried out
in time proportional to the number of experts considered.
Linear time convolution is possible because the kernel is a memoryless distribution conditional on the sign of the drift. This sign information is represented by the distinction between two columns of silent states for each time step.
\begin{figure}
\caption{Parameter drift: $\ehmm{pd}[\alpha]$}\label{fig:geometric.kernel.update}
\centering
\subfloat{%
$\xymatrix@!0@R=2em@C=2.1em{
&&&\ar@{.>}[d]&&&\ar@{.>}[d]\\
&&\si \ar@{.>}[u]&\si \ar[dd] \ar[dr]&&\si \ar@{.>}[u]&\si \ar[dd] \ar[dr]\\
&\ex{2} \ar[ur] \ar[rrd] \ar[rrr]&&& \ex{2} \ar[ur] \ar[rrd] \ar[rrr] &&& \ex{2}\\
&&\si \ar[uu] \ar[rru]&\si \ar[dd] \ar[dr]&&\si \ar[uu] \ar[rru]&\si \ar[dd] \ar[dr]\\
&\ex{1} \ar[ur] \ar[rrd] \ar[rrr] &&& \ex{1} \ar[ur] \ar[rrd] \ar[rrr] &&& \ex{1}\\
&&\si \ar[uu] \ar[rru]&\si \ar[dd] \ar[dr]&&\si \ar[uu] \ar[rru]&\si \ar[dd] \ar[dr]\\
\bn\ar@{.>}[uuuuuur]\ar[uuuur]\ar[uur]\ar[r]\ar[ddr]\ar[ddddr]\ar@{.>}[ddddddr]&\ex{0}
\ar[ur] \ar[rrd] \ar[rrr] &&& \ex{0} \ar[ur] \ar[rrd] \ar[rrr] &&& \ex{0} & \cdots\\
&&\si \ar[uu] \ar[rru]&\si \ar[dd] \ar[dr]&&\si \ar[uu] \ar[rru]&\si \ar[dd] \ar[dr]\\
&\ex{-1} \name[+<3pt,13pt>]{1,-1} \ar[ur] \ar[rrd] \ar[rrr]  &&& \ex{-1} \name[+<1pt,13pt>]{2,-1}\ex{-1} \ar[ur] \ar[rrd] \ar[rrr]  &&& \ex{-1}\\
&&\si \name[+<-14pt,0pt>]{1,-2,-1}\ar[uu] \ar[rru]&\si \name[+<14pt,0pt>]{1,-1,-2} \ar[dd] \ar[dr]&&\si \ar[uu] \ar[rru]&\si \ar[dd] \ar[dr]\\
&\ex{-2} \ar[ur] \ar[rrd] \ar[rrr]  &&& \ex{-2} \ex{-2} \ar[ur] \ar[rrd] \ar[rrr]  &&& \ex{-2}\\
&&\si \ar[uu] \ar[rru]&\si \ar@{.>}[d] &&\si \ar[uu] \ar[rru]&\si \ar@{.>}[d]\\
&&\ar@{.>}[u]&&&\ar@{.>}[u]&\\
}$}
\quad
\subfloat{%
\setlength{\abovedisplayskip}{0cm}
\setlength{\belowdisplayskip}{0cm}
\small
\begin{minipage}[t]{.5\textwidth}
\begin{gather*}
\ehmm{pd}[\alpha]=\tuple{\wstates, \wstates^\prd, \wstart, \wtf,
  \wnl}\qquad
\wstates = \wstates^\sil \cup \wstates^\prd\\
\wstates^\prd = \posints\times\ints\quad
\wstart = 0\quad
\wnl(t, \xi) = \xi 
\\
\wstates^\sil=\posints\times\{\tuple{i,i+1},\tuple{i,i-1}\mid i\in\ints\} \; \cup \; \set{0}
\\
\wtf\del*{
\begin{aligned}
\tuple{0} &\to \tuple{1,\xi}\\
\tuple{t,\xi\!-\!1,\xi} &\to \tuple{t,\xi,\xi\!+\!1}\\
\tuple{t,\xi\!+\!1,\xi} &\to \tuple{t,\xi,\xi\!-\!1}\\
\tuple{t,\xi\!-\!1,\xi} &\to \tuple{t\!+\!1,\xi}\\
\tuple{t,\xi\!+\!1,\xi} &\to \tuple{t\!+\!1,\xi}\\
\tuple{t,\xi}&\to\tuple{t\!+\!1,\xi}\\
\tuple{t,\xi}&\to\tuple{t,\xi,\xi\!+\!1}\\
\tuple{t,\xi}&\to\tuple{t,\xi,\xi\!-\!1}
\end{aligned}} = \del*{
\begin{gathered}
\kappa_\alpha(\xi)\\
\alpha\\
\alpha\\
1-\alpha\\
1-\alpha\\
(1\!-\!\alpha)/(1\!+\!\alpha)\\
\alpha/(1+\alpha)\\
\alpha/(1+\alpha)
\end{gathered}}
\end{gather*}
\end{minipage}
}
\end{figure}


\begin{theorem}[Parameter Drift Regret] \label{thm:oe}
  Fix any data $x^t$ and reference sequence $\xi^t$ with total drift
  $d$. Let $\ent(P,Q)=-\sum_x P(x)\ln Q(x)$ denote the cross entropy. Then
\[
\ln\frac{\ehmm{P}[\xi^t](x^t)}{\ehmm{pd}[\alpha](x^t)}
~\le~
t \ent(\kappa_{\alpha^*}, \kappa_{\alpha})
~=~
-t\ln\frac{1-\alpha}{1+\alpha}-d\ln\alpha,
\]
where $\alpha^* = \argmax_\alpha \ehmm{pd}[\alpha](\xi^t)= \sqrt{1 + (t/d)^2} - (t/d)$.
\end{theorem}

\begin{IEEEproof} By \lemref{lem:dropmarg}, the left-hand side is bounded above by $- \ln \ehmm{pd}[\alpha](\xi^t)$.
Since $\set{\kappa_\alpha}$ is an exponential family with unit carrier (see e.g.\ \cite[Proposition 19.1]{grunwald2007}),
\[\begin{split}
- \ln \ehmm{pd}[\alpha](\xi^t)
&=
- \ln \prod_{i=1}^t \kappa_\alpha(\delta_i)\\
&=
t \Exp_{\kappa_{\alpha^*}} \sbr*{- \ln \kappa_\alpha(\rv\delta)}
=
t \ent(\kappa_{\alpha^*}, \kappa_\alpha).
\end{split}\]
The equality of the Theorem further follows from
\[
\ehmm{pd}[\alpha](\xi^t)
~=~\prod_{i=1}^t  \kappa_\alpha(\delta_i)
~=~\alpha^d\del*{\frac{1-\alpha}{1+\alpha}}^t.
\]
The parameter $\alpha^*$ that maximises the likelihood of $\xi^t$ is
found by equating the derivative to zero.
\end{IEEEproof}

We can be somewhat more precise about how much it can hurt performance
to use a suboptimal parameter $\alpha$. The following theorem, which
bounds the regret with respect to the optimal parameter-drift model,
is an analogue of Theorem~\ref{thm:fs.parambound} for Fixed Share. The
theorem applies to a wide class of kernel EHMMs, but in particular it
holds for the parameter-drift model $\ehmm{pd}[\alpha]$, for which the transition dynamics are governed by the one-dimensional exponential family \eqref{eq:kappa.alpha}. It is a strong result that uses the full generality of \lemref{lem:parambound}.

\begin{theorem}[Kernel ML Regret]\label{thm:kernel_ml_regret}
Fix $x^t$ and let $\hat\beta=\argmax_\beta
\ehmm{kernel}[\kappa_\beta](x^t)$ for some exponential family $\set{\kappa_\beta}$. We have
\[
\begin{split}
\ln\frac{\ehmm{kernel}[\kappa_{\hat\beta}](x^t)}{\ehmm{kernel}[\kappa_\beta](x^t)}
~\le~
t \kld\delcc[\big]{\kappa_{\hat\beta}}{\kappa_\beta}.
\end{split}
\]
\end{theorem}
\begin{IEEEproof}
  Since the transition probabilities
  associated with each productive state (i.e.\ the kernel
  $\kappa_\beta$) are an exponential family distribution, we can apply
  \lemref{lem:parambound} with $\Qpara$ equal to the set of all
  productive states.
\end{IEEEproof}
Instantiating this theorem for parameter drift we obtain:
\begin{corollary}[Parameter Drift ML Regret]\label{crlr:pd.ml}
Fix $x^t$ and let $\hat\alpha=\argmax_\alpha \ehmm{PD}[\alpha](x^t)$. We have
\[
\ln\frac{\ehmm{PD}[\hat\alpha](x^t)}{\ehmm{PD}[\alpha](x^t)}
~\le~
t
\del*{
\frac{
  2 \hat\alpha \ln \left(\frac{\hat\alpha}{\alpha}\right)
}{
  (1\!-\!\hat\alpha)(1\!+\!\hat\alpha)
}
+
\ln \left(\frac{(1+\alpha) (1-\hat\alpha)}{(1-\alpha) (1+\hat\alpha)}\right)
}
.
\]
\end{corollary}

The parameter drift model as discussed so far shares both the elegance
of the Fixed Share algorithm and its awkward dependence on a 
parameter $\alpha$. However, most of the techniques to avoid specifying
$\alpha$ that were discussed in previous sections can be adapted to
the parameter drift model. In particular, we can compete with the maximum likelihood drift parameter $\hat \alpha$ using a discretisation scheme akin to \cite{Jaakkola2003,DeRooijVanErven09}, such that the discretisation point $\alpha$ closest to $\hat \alpha$ reduces the  right hand side of Corollary~\ref{crlr:pd.ml} to a uniformly bounded constant. As before, this is possible using $O(\sqrt{t})$ discretisation points, leading to a $O(t \sqrt{t})$ total running time. We omit the details.

Moreover we can adapt the trick we used in
\secref{sec:dsr}, and let the kernel parameter $\alpha$
decrease with time. This yields a linear run time, at the cost of deteriorating the bound for large drifts.
\begin{theorem}[Decreasing Drift Regret]
  Let $\ehmm{pd}[]$ denote the ES-joint based on the
  parameter drift model with time-dependent kernel $\kappa_{\alpha_i}$
  with $\alpha_i=1/(i+1)$. For any data $x^t$ and reference sequence
  $\xi^t$ with total drift $d$, we have
  \[
  \ln\frac{\ehmm{P}[\xi^t](x^t)}{\ehmm{pd}[](x^t)}~\le~
  (d+2)\ln(t+1).
  \]
\end{theorem}
\begin{IEEEproof} We first expand
  \begin{multline*}
    \ehmm{pd}[](\xi^t)=\prod_{i=1}^t
    \kappa_{\alpha_i}(\delta_i)=\prod_{i=1}^t
    \alpha_i^{\abs{\delta_i}}\frac{1-\alpha_i}{1+\alpha_i}\\
    =\prod_{i=1}^t (i+1)^{-\abs{\delta_i}}\frac{i}{i+2}=\frac{2}{(t+1)(t+2)}\prod_{i=1}^t(i+1)^{-\abs{\delta_i}}.
  \end{multline*}
  For fixed total drift $d$, it is clear that this probability
  is minimised by $\abs{\delta_i}=0$ for $1\le i<t$ and $\abs{\delta_t}=d$.
  Therefore
  \[
    \ehmm{pd}[](\xi^t)\ge\frac{2}{(t+1)(t+2)}(t+1)^{-d}\ge
    (t+1)^{-d-2}.
  \]
  We now take the $-\ln$ and apply \lemref{lem:dropmarg} to complete the proof.
\end{IEEEproof}

\section{Extensions}\label{sec:loose.ends}
In this section we describe a number of extensions to the framework
described above. In \secref{sec:conditioning.on.data} we discuss a
number of tracking algorithms for which different, potentially useful
performance guarantees can be given. Then in \secref{sec:adaptive} we discuss \emph{adaptive regret}, a criterion for evaluating performance more locally. In \secref{sec:map} we try to
find out which expert was best at a particular time step, and finally
in Sections~\ref{sec:loss.functions} and~\ref{sec:investment} and we indicate how our approach can be
generalised to work with any mixable loss function and how it can be applied to online investment.

\subsection{Other Approaches to
  Switching}\label{sec:conditioning.on.data}
The models described in the present paper can be described in the
Bayesian framework using prior distributions on sequences of
experts. The priors we presented so far did not depend on any contextual information, such
as the outcomes, or any other external information. Below, we will
list three other important models for expert tracking whose EHMM
transition probabilities depend on the past losses of the experts. We
are not aware of any models in which the transition probabilities
depend on properties of the observed data other than the losses; this
is an interesting area for future research.

In all three cases it is straightforward to find counterexamples that show that these algorithms cannot be represented as EHMMs with fixed transition probabilities.

The first such algorithm, Variable Share, was introduced together with
Fixed Share in \cite{HerbsterWarmuth1998}. It is useful for loss
functions like square loss, that are not only mixable, but also
bounded (also see Section~\ref{sec:loss.functions}). For this setting,
which is outside the scope of this paper, a regret of $O(m\ln
k+m\ln(L/m))$ can be established, where $m$ and $k$ are the number of
blocks and the number of experts as usual, but $L$ is the loss of the
best reference strategy using $m$ blocks. Thus, if the data can be
predicted well by partitioning it into blocks and using a fixed expert
within each block, then the overhead of the algorithm is small. In the
log loss setting however, the algorithm incurs infinite regret in the
worst case, so no useful guarantees can be provided.

The other two approaches do in fact work in the log loss setting. The
first is known as ``Mixing Past Posteriors'' \cite{bousquet2002}. Like
Fixed Share, this algorithm allows for efficient tracking of the best
predicting expert; but unlike Fixed Share, it is especially efficient
for sparse problems, where the predictions of only a few out of the
full set of experts ever need to be used. For Mixing Past Posteriors,
a regret of $O(u\ln k+m\ln t)$ can be achieved, where $u$ is the
number of experts that occur in the comparator sequence, and $t$, $k$
and $m$ are as usual. This is beneficial if $k$ is large and $u$ is
small. Whereas the original algorithm is rather eclectic, a proper
Bayesian interpretation (making use of specialists) and a slight improvement of the bound can be
found in \cite{bayes.sleep}.

The last result combines two experts in such a way that the regret is
controlled in terms of the fluctuations in the cumulative loss
difference of the two experts as a function of time. The idea is that
if the fluctuations are large, the regret is relatively high, but in
that case you also gain a lot from switching between the experts in
the first place. The paper \cite{KoolenDeRooij2013} is phrased in
terms of investment policies, but the setting is equivalent to ours. In
financial terms, the bound expresses that you have a large overhead
only when you are making a lot of money anyway!

\subsection{Adaptive Regret}\label{sec:adaptive}
Also of interest is the notion of \emph{adaptive regret} proposed by Hazan and Seshadhri in \cite{hazan2009efficient}. The adaptive regret of an algorithm on a given interval is the difference between the loss of the algorithm on that interval and that of the best expert for that interval. The new goal is then to design algorithms with low worst-case adaptive regret on all intervals. An algorithm with low adaptive regret will automatically have low tracking regret, the tracking bound is obtained simply by summing the adaptive bound over all blocks of a segmentation of the data; the converse is not always possible. It is proved in \cite{adaptive.regret} that Fixed Share, and its generalisations with  time-varying switching rates (as e.g.\ in Section~\ref{sec:dsr}) are optimal algorithms for adaptive regret: no other algorithm can guarantee lower adaptive regret on all intervals. Moreover in that paper's forthcoming journal version \cite{adaptive.regret.journal} it is shown that the worst-case adaptive regret of any algorithm is dominated by that of such a generalised Fixed Share. As such, whereas the local perspective taken by adaptive regret allows giving stronger performance guarantees for Fixed Share, it cannot capture the global benefit of modelling the switching dynamics, as expressed e.g.\ by the bounds for run-length and the Switching Method.

\subsection{Expert Estimation}\label{sec:map}
We focused on EHMM models for sequential prediction. However, EHMMs may also be used for batch data analysis. Below we indicate how to obtain the best regularised expert sequence, for example to gauge change-points in the data. We then discuss calculating the posterior marginal distribution on experts at each time step. This can be used to visualise the evolution of the prediction performance of each of the experts.
The forward algorithm computes the probability of the data, that is
\[
\Q(x^t) 
~=~
\sum_{\wstate^{(t)}} \Q(x^t,\wstate^{(t)}),
\]
Instead of the entire sum, we are sometimes interested in the sequence
of states $\wstate^{(t)}$ that contributes most to it:
\[
\argmax_{\wstate^{(t)}} \Q(x^t, \wstate^{(t)}) ~=~ \argmax_{\wstate^{(t)}} \Q(x^t|\wstate^{(t)}) \Q(\wstate^{(t)}).
\]
The Viterbi algorithm~\cite{rabiner1989} is used to compute the most
likely sequence of states for HMMs. It can be easily adapted to handle
silent states. However, we may also write
\[
\Q(x^t)
~=~
\sum_{\xi^t} \Q(x^t,\xi^t),
\]
and wonder about the sequence of \emph{experts} $\xi^t$ that
contributes most. This problem is harder because several states can
produce the same expert simultaneously; in other words, a single
sequence of experts can be generated by many different sequences of
states. So we cannot use the Viterbi algorithm as it is. The Viterbi
algorithm can be extended to compute the MAP expert sequence for
general EHMMs, but the resulting running time explodes.
Still, the MAP $\xi^t$ can be sometimes be obtained efficiently by
exploiting the structure of the EHMM at hand. This turns out to be
possible for Fixed Share, and also for the more sophisticated Switch
Distribution mentioned in \secref{sec:switch.fast}; the algorithm for
the latter is given in \cite{us:ceae:corr}.

As an alternative way to gain insight, one may run the forward and
backward algorithms to compute $\Q(x^i, \wstate^\prd_i)$ and $\Q(x^t |
\wstate^\prd_i, x^i)$. Recall that $\wstate^\prd_i$ is the productive
state that is used at time $i$.  From these we can compute the
a posteriori probability $\Q(\wstate^\prd_i | x^t)$ of each productive
state $\wstate^\prd_i$.  That is, the posterior probability taking all
the available data into account (including observations that were made
later than time $i$). This is a standard way to analyse data in the
HMM literature, see e.g. \cite{rabiner1989}. We can then project the
posterior on \emph{states} down to obtain the posterior probability
$\Q(\xi_i | x^t)$ of each \emph{expert} $\xi_i \in \Xi$ at each time
$i=1, \ldots, t$.  This gives us a sequence of mixture weights over
the experts that we can, for example, plot on a $\Xi \times t$ grid.
On the one hand this gives us a mixture over experts for each time
instance, obviously a richer representation than just single experts.
On the other hand we lose the temporal correlations that can be
important in MAP calculation, as each time instance is treated
separately.

\subsection{Mixable Loss Functions}\label{sec:loss.functions}
We presented log-loss regret bounds for experts that sequentially produce probability distributions on the next outcome. Not all prediction tasks are in this form, for example, we may be asked to make a point prediction based on real-valued expert advice and be scored using quadratic loss. Fortunately, several loss functions are \markdef{mixable}~\cite{cesa-bianchi2006,haussler1998}, in that for each mixture of predictions, there is a single prediction whose loss is always less than the exponentiated average loss. Mixable losses include log loss, quadratic loss, Hellinger loss and entropic loss. 0/1 loss and absolute loss are not mixable.

Prediction strategies that are obtained by running the forward
algorithm on any EHMM can be adapted to mixable losses
straightforwardly, by preprocessing the input to and post-processing
the output of the forward algorithm for sequential prediction. On the
input side, expert predictions are transformed into probabilities. On
the output side, the posterior distribution on the next expert is
transformed (using the mixability condition) into a single
prediction. The resulting prediction strategy has the \emph{same}
mixable-loss regret bound as the original prediction strategy
(although possibly expressed in different units).
The details of the general reduction can be found
in~\cite{koolen09:_fixed_share_for_learn_exper}. In the special case of online investment an even tighter correspondence holds, as outlined in the next section.


\subsection{Online Investment}\label{sec:investment}

As it happens, all algorithms for prediction with expert advice discussed in this paper can also be used as strategies for online
investment. The key observation is that the weights on the experts issued by any
considered prediction algorithm and the resulting codelength
\emph{only} depend on the \emph{losses} incurred by the experts, not
on any other aspect of their behaviour. This is clear from the motivating definition \eqref{eq:ES.joint}.

To predict on the stock market, we start again from \eqref{eq:ES.joint}, but replace expert $\xi$'s data likelihood $P_{\xi,s}(x_s)$ with the multiplication factor $r_{\xi,s}$ incurred by  stock $\xi$ in trading round $s$. Each round, the investment algorithm uses the same expression to compute the posterior distribution on the next stock. It then divides its capital among the stock according to this distribution. In the prediction algorithm this posterior was used instead to mix the predictions of the experts to form its own predictive distribution; for the investment strategy this last step has been abstracted away.

Note that while the expert data likelihoods are proper probability distributions, the multiplication factors $r_{\xi,t}$ may be larger than one; this does not cause problems since the posterior weights are renormalised by Bayes' rule.

The returns obtained by the investment strategy are given by the formula for the data likelihood with the above substitution. As such, all regret guarantees derived in this paper carry over unmodified when using these algorithms as investment strategies.

For further discussion of the relation between investment and prediction algorithms see \cite[Chapter 10]{cesa-bianchi2006}.

\subsubsection*{Cover's Portfolio Selection}
Perhaps the best known link between information theory and finance is provided by Cover's seminal results on portfolio selection~\cite{cover1991universal,Cover2006}. These algorithms fit exactly in the formalism described in this paper: they can be obtained by applying the reduction described above to specific EHMM models. The simplest such model was introduced in Example~\ref{example:fixed.elementwise.mixtures}, where each round, the predictions of the experts are mixed using a fixed weight vector. Applying the reduction to finance, we recover the \emph{constant rebalanced portfolio} strategy. It is also possible to obtain Cover's \emph{universal} portfolios by using a more sophisticated EHMM that learns the optimal mixture weights. For the case of two experts, this EHMM has already been defined, albeit for a different purpose: the EHMM depicted in \figref{graph:universal.share} was previously used as an interpolator, for learning the switching rate in an expert tracking strategy. However, when applied in its own right it learns the optimal elementwise mixture weights for combining the predictions of two experts labelled ``\swi'' and ``\nsw''. 

The construction of~\figref{graph:universal.share} can be augmented to more than two experts, but the state space quickly grows large: for $k$ experts, the number of states in round $t$ is $t^{k-1}$. As such, the algorithm will process $t$ outcomes in $O(t^k)$ time, matching the complexity of Cover's algorithm. Interestingly, the methods discussed in~\secref{sec:montjaak} for reducing the time complexity of the Switching Method carry over to learning mixtures, allowing an easy speedup to $O(t^{(k+1)/2})$.

Substantial advances have been made in making Cover's universal portfolio selection practical for large numbers of stocks (by imposing some assumptions); these fall outside the scope of this paper. For more information see~\cite{HazanAgarwalKale2007}.

\section{Conclusion}
We generalise the concept of universal coding for some model class
$\cal M$, by comparing the performance of the universal code not just
to the performance of the codes in $\cal M$, but also to other
reference classes. 

We evaluate performance in terms of individual sequence regret, and
make no distributional assumptions. We summarise and unify existing algorithms from two domains: information-theoretic literature about
universal coding on the one hand and universal prediction (also known as
``prediction with expert advice'') from learning theory. Thanks to the well-known
equivalence between prefix coding and probability theory the algorithms and techniques of this paper can immediately be applied in both settings.

We present all models in Bayesian form using prior
distributions on expert sequences (ES-priors). The (infinitely long)
expert sequence defines which expert is used at which time. Prediction
then amounts to ``integrating out'' those experts in the sequence that
are used at other time steps than the one predicted. The challenge is
to identify those models that provide good tradeoffs between
predictive performance and time complexity.

Throughout the paper, hidden Markov models (HMMs) are used to specify ES-priors, since their
explicit representation of the current state and state-to-state
evolution naturally fit the temporal correlations we seek to model.
For reasons of efficiency we use HMMs with silent states. The standard
algorithms for HMMs (Forward, Backward, Viterbi and Baum-Welch) can be
used to answer questions about the ES-prior as well as the induced
distribution on data. The running time of the forward algorithm can be
read off directly from the graphical representation of the HMM.

This approach allows a unified presentation of many existing expert models. We
focus on models for \emph{tracking the best expert}, where the loss
incurred by a prediction strategy is compared to the loss incurred if
the data are optimally divided into $m$ blocks, and the best expert is
used within each block. The discrepancy (``regret'') is then bounded
in terms of variables such as the current time $t$, the number of
experts $k$, and the number of blocks $m$. In each case, we recover
both the regret bound and the running time known
from the literature.

We not only succinctly summarise and
contrast many key algorithms from the literature, but also describe
a number of new models. In particular the model with quickly decreasing
probability of switching (\secref{sec:dsr}) and the models that assume the experts are ordered
(\secref{sec:oe}) are new and computationally efficient, and have
competitive regret bounds.

\section*{Acknowledgements}
Peter Gr\"unwald's and Tim van Erven's suggestions significantly
improved this paper. Thanks also go to Mark Herbster for an enjoyable
afternoon exchanging ideas, which has certainly influenced the shape
of this paper. We thank Wojciech Kot\l owski and Thijs van Ommen for proofreading. 

\bibliographystyle{IEEEtran}
\bibliography{experts}

\begin{IEEEbiography}[{\includegraphics[width=1in,height=1.25in,clip,keepaspectratio]{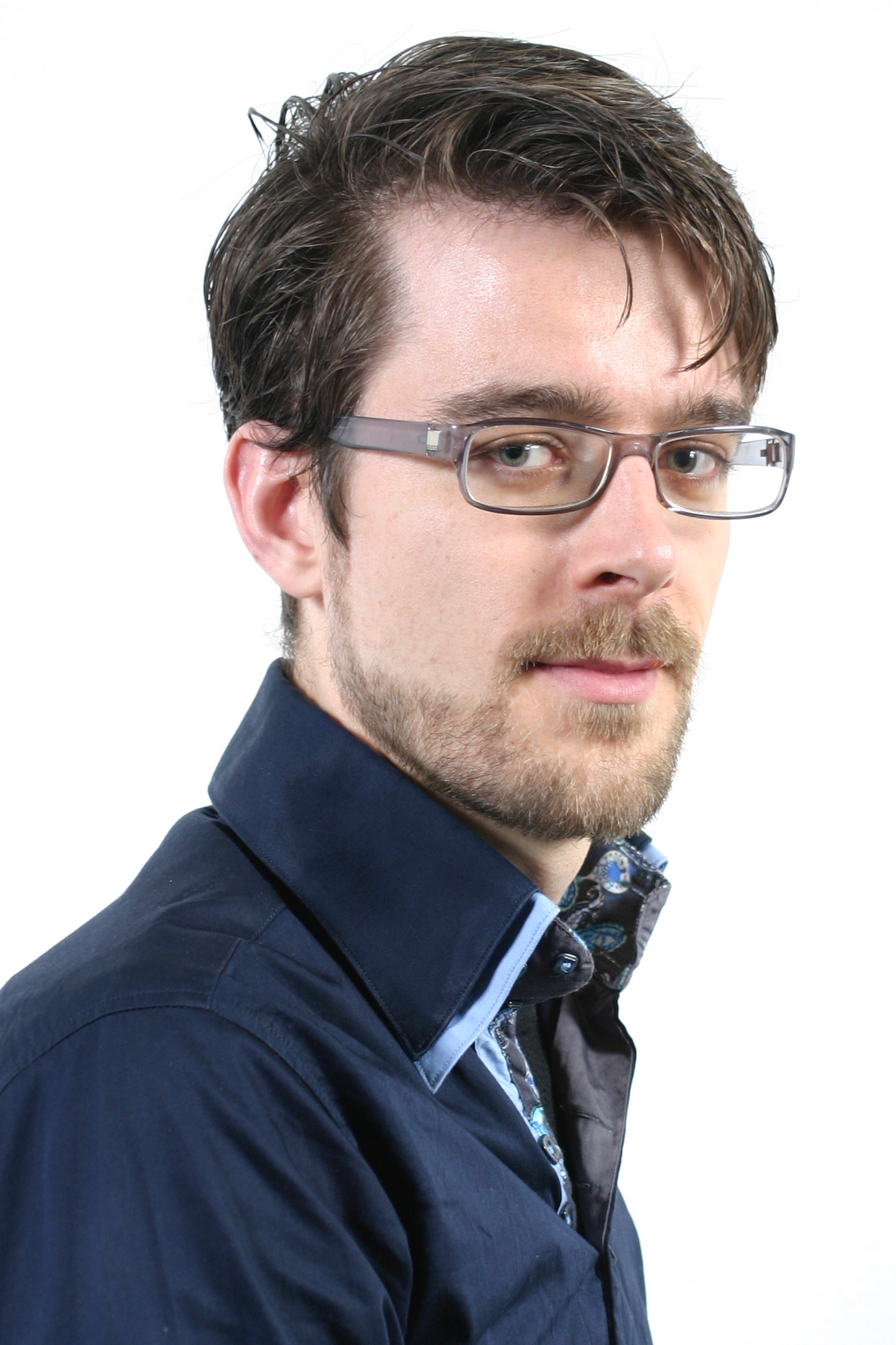}}]{Wouter M. Koolen}
carried out his graduate research at CWI, supervised by Professor Paul Vit\'anyi and
  Professor Peter Gr\"unwald. He graduated cum laude in 2011 at the University of Amsterdam. He then took up a research fellowship at the Computer Learning Research Centre at Royal Holloway, University of London. He joined the information theoretic learning group at CWI in February 2013. His interests are online decision making, minimax algorithms, Bayesian reasoning and finance.
\end{IEEEbiography}
\begin{IEEEbiography}[{\includegraphics[width=1in,height=1.3in,clip,keepaspectratio]{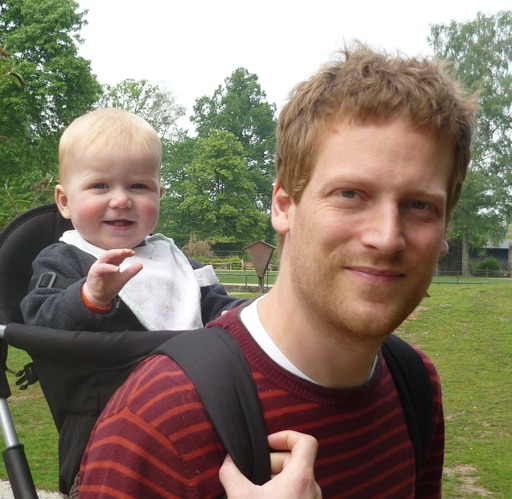}}]{Steven
    de Rooij}
  received his PhD in 2008 from the University of Amsterdam, under
  supervision of Professor Paul Vit\'anyi and Professor Peter
  Gr\"unwald at CWI institute in Amsterdam, Netherlands. From 2008 to 2010 he
  was Research Associate in the Statistical Laboratory at the
  University of Cambridge; he currently works at the University of
  Amsterdam and VU university, where he applies concepts from Minimum
  Description Length learning, Bayesian inference, model selection,
  and sequential prediction to the semantic web.
\end{IEEEbiography}

\end{document}